\documentclass[preprint, authoryear]{elsarticle}
\usepackage[OT1]{fontenc}
\usepackage[utf8]{inputenc}
\usepackage{amsmath}
\usepackage{enumerate}
\usepackage{amsfonts}
\usepackage{color}
\usepackage{amssymb}
\usepackage{amsthm}
\usepackage{dsfont}
\usepackage{algorithmic}
\usepackage[authoryear]{natbib}
\usepackage{hyperref}
\usepackage[hypcap]{caption}
\usepackage{graphicx}
\usepackage{here}
\usepackage{booktabs}
\usepackage{algorithm2e}
\usepackage{framed}
\usepackage{Sweave}
\usepackage{tabularx}

\numberwithin{equation}{section}

\newtheoremstyle{myprop}
  {6pt}
  {3pt}
  {\itshape}
  {}
  {\scshape}
  {.}
  {.5em}
  {}
\theoremstyle{myprop}
\newtheorem{proposition}{Proposition}[section]
\newtheorem{lemma}{Lemma}

\newtheorem{definition}[proposition]{Definition}

\theoremstyle{remark}

\newtheoremstyle{example}
  {6pt}
  {3pt}
  {\small\sffamily}
  {}
  {\small\sffamily\slshape}
  {.}
  {.5em}
  {}
\theoremstyle{example}

\renewcommand{\ge}{\geqslant}

\newcommand{\lca}{\operatorname{lca}}

\begin{document}
\title{Nonparametric estimation of the tree structure of a nested Archimedean copula}

\author[belgien]{Johan Segers\fnref{ARC}}
\ead{johan.segers@uclouvain.be}
\author[belgien]{Nathan Uyttendaele\fnref{ARC}\corref{cor1}}
\ead{nathan.uyttendaele@uclouvain.be}

\address[belgien]{Universit\'e catholique de Louvain, Institut de Statistique, Biostatistique et Sciences Actuarielles, Voie du Roman Pays 20, B-1348 Louvain-la-Neuve, Belgium}

\fntext[ARC]{Order of contributions not necessarily reflected by alphabetical order.}
\cortext[cor1]{Corresponding author}

\begin{abstract}
One of the features inherent in nested Archimedean copulas, also called hierarchical Archimedean copulas, is their rooted tree structure. A nonparametric, rank-based method to estimate this structure is presented. The idea is to represent the target structure as a set of trivariate structures, each of which can be estimated individually with ease. Indeed, for any three variables there are only four possible rooted tree structures and, based on a sample, a choice can be made by performing comparisons between the three bivariate margins of the empirical distribution of the three variables. The set of estimated trivariate structures can then be used to build an estimate of the target structure. The advantage of this estimation method is that it does not require any parametric assumptions concerning the generator functions at the nodes of the tree.
\end{abstract}

\begin{keyword}
Archimedean copula \sep 
dependence \sep
nested Archimedean copula \sep 
hierarchical Archimedean copula \sep 
rooted tree \sep
subtree \sep 
Kendall distribution \sep 
fan \sep
triple \sep
nonparametric inference
\end{keyword}
\maketitle
\section{Introduction}
Archimedean copulas have become a popular tool for modeling or simulating bivariate data. They are however not useful for every application, failing for instance to properly model in high dimensions if the data do not exhibit
symmetric dependencies. Nested Archimedean copulas (NACs), or hierarchical Archimedean copulas, are an interesting attempt to overcome this problem. They were first introduced by \citet[pp.~87--89]{tJOE97a} and then have been studied extensively, see for instance \cite*{doi:10.1080/00949650701255834}, \cite*{hofert2008sampling}, \cite*{hofert2010sampling} or \cite*{hofert2011efficiently} for sampling algorithms; \cite*{Hofert:Maechler:2010:JSSOBK:v39i09}, who released the first \textsf{R} package devoted to NACs; \cite*{hering2010constructing}, who investigated the construction of NACs with Lévy subordinators; \cite*{hofert2012densities}, who examined their densities; \cite*{OOW}, who were the first to investigate likelihood-based estimation; or \cite{okhrinOstap2013properties}, who studied tail dependence properties of NACs.

The hierarchy of variables in a nested Archimedean copula is described through a rooted tree. Most often, the tree is given from the context; see for instance \cite*{hofert2010sampling} or \cite*{puzanova2011hierarchical}. \cite*{OOW} were the first to address the issue of reconstructing the tree from a sample, offering a parametric answer in which each generator is assumed to be known up to some Euclidean parameter(s). In contrast, the method we propose is completely nonparametric since it does not require the user to make any assumption about the generators of the NAC from which the tree structure must be estimated, except for a rather straightforward identifiability condition introduced in Section \ref{identifiability_sec}. Although never formally mentioned, this identifiability condition is assumed throughout \cite*{OOW} as well.

Sections \ref{archi_sec} and \ref{nac_sec} of this paper review the basics of Archimedean copulas and nested Archimedean copulas. Section \ref{test} shows how the structure of three variables $(X_i,\allowbreak X_j,\allowbreak X_k)$ can be estimated nonparametrically. The idea is to estimate the Kendall distribution associated with each pair of variables within $(X_i,\allowbreak X_j,\allowbreak X_k)$; these estimates allow us then to decide if all pairs of variables have actually the same underlying bivariate distribution or not. If so, then the tree structure of $(X_i,\allowbreak X_j,\allowbreak X_k)$ is the trivial trivariate structure, that is, a structure with one internal vertex and three leaves, also called a 3-fan. If not, determining which pair has a different underlying bivariate distribution allows one to assign the correct tree structure to $(X_i,\allowbreak X_j,\allowbreak X_k)$.

Section \ref{sufficiency_triple} introduces a key point, namely that a given tree structure $\lambda$ can always be represented as a set of trivariate structures. That is, for a random vector of continuous random variables $(X_1, \ldots, X_d)$ with a nested Archimedean copula, it is possible to obtain the tree structure of this nested Archimedean copula provided the tree structure of the nested Archimedean copula associated with any three variables $(X_i, X_j, X_k)$ with distinct $i, j, k \in \{1, \ldots d \}$ is known. A very similar result was obtained by \cite*{ng1996}, who showed that a given structure $\lambda$ can always be represented as a set of \emph{triples} and \emph{fans}, triples and fans being formally defined in Section~\ref{sufficiency_triple}. Another interesting result is offered by \cite*{okhrinOstap2013properties} who showed that the structure can be retrieved from the bivariate margins of the nested Archimedean copula.

Our suggestion to estimate the structure of $(X_1, \ldots, X_d)$ is first to estimate the tree structure of the nested Archimedean copula associated with any three variables $(X_i, X_j, X_k)$ with distinct $i, j, k \in \{1, \ldots d \}$, and second to use this set of estimated trivariate structures to build an estimate of the structure of $(X_1, \ldots, X_d)$. This suggestion and one important related difficulty make up Section \ref{problem}.

The performance of our approach is then assessed by means of a simulation study involving target structures in several dimensions (Section \ref{sim_sec}). As part of this simulation study, the performance of the approach used by \cite*{OOW} is also investigated. 

Finally, Section \ref{app_sec} illustrates how our method could be used to highlight hierarchical interactions in the stock market. Some remaining challenges are outlined in Section \ref{dis_sec}.

\section{Archimedean copulas \label{archi_sec}}
Let $(X_1, \dots, X_d)$ be a vector of continuous random variables. The copula of this vector is defined as
\begin{equation} \nonumber
C(u_1, \dots, u_d)=P(U_1\leq u_1, \dots, U_d\leq u_d),
\end{equation}
where $(U_1, \dots, U_d)=(F_{X_1}(X_1), \dots, F_{X_d}(X_d))$, and where $F_{X_1}, \dots, F_{X_d}$ are the marginal cumulative distribution functions (CDFs) of $X_1, \dots, X_d$, respectively.

Archimedean copulas are a class of copulas that admit the representation
\begin{equation} \nonumber
C(u_1, \dots, u_d)=\psi(\psi^{-1}(u_1) +\dots+\psi^{-1}(u_d)),
\end{equation}
where $\psi$ is called the generator and $\psi^{-1}$ is its generalized inverse, with $\psi : [0, \infty) \rightarrow [0, 1]$, a convex, decreasing function such that $\psi(0)=1$ and $\psi(\infty)=0$. In order for $C$ to be a $d$-dimensional copula, the generator is required to be $d$-monotone on $[0, \infty)$, see \cite*{2009arXiv0908.3750M} for details.

The generators in Table \ref{archi_families} are among the most popular ones. All of them are completely monotone, that is, $d$-monotone for all integer $d \geq 2$. For the Frank family, $D_1(\theta)=\frac{1}{\theta}\int_0^\theta t/(\exp(t)-1)\,dt$.

\begin{table}[H]
  \centering
  \caption{Some popular generators of Archimedean copulas. \label{archi_families}}
  \scriptsize
    \begin{tabular}{cccc}
    \toprule
    name & generator $\psi(x)$   & $\theta $&$ \tau$\\
    \midrule
    
    AMH & $(1-\theta)/(e^x-\theta)$  & $\theta \in [0, 1)$ & $1-2\left(\theta+(1-\theta)^2 \log(1-\theta)\right)/(3\theta^2)$\\
       Clayton & $ (1+x)^{-1/\theta}$  & $\theta \in (0, \infty)$ &  $\theta/(\theta+2)$\\
       Frank & $ -\log(1-(1-e^{-\theta})e^{-x})/\theta$  & $\theta \in (0, \infty)$ & $1+4(D_1(\theta)-1)/\theta$\\
      Gumbel & $ \exp(-x^{1/\theta}) $  & $\theta \in [1, \infty)$ & $(\theta-1)/\theta$\\
          Joe & $ 1-(1-e^{-x})^{1/\theta}$ & $\theta \in [1, \infty)$ & $1-4\sum_{k=1}^\infty 1/(k(\theta k+2)(\theta(k-1)+2))$ \\

	\bottomrule
    \end{tabular}
\end{table}

The parameter $\theta$ in Table \ref{archi_families} allows one to control the strength of the dependence between any two variables of the related Archimedean copula. This is best understood by expressing Kendall's $\tau$ coefficient between any two variables of the related Archimedean copula in terms of $\theta$ \citep{Hofert:Maechler:2010:JSSOBK:v39i09}, as done in the last column of Table \ref{archi_families}.

All margins of the same dimension of an Archimedean copula are equal, that is, for all $m \in \{ 2, \ldots, d \}$ and for every subset $\{ i_1,
\ldots, i_m \}$ of $\{ 1, \ldots, d \}$ having $m$ elements, the two vectors
\begin{equation} \nonumber
(U_{i_1}, \ldots, U_{i_m}) \text{ and } (U_{1}, \ldots, U_{m}) \\
\end{equation}
have the same distribution. This result stems from the fact that for Archimedean copulas, $C(u_1, \dots, u_d)$ is a symmetric function in its arguments and this is why Archimedean copulas are sometimes also called \emph{exchangeable}. For modeling purposes, this exchangeability becomes
an increasingly strong assumption as the dimension grows.

\section{Nested Archimedean copulas \label{nac_sec}}

Asymmetries, allowing for more realistic dependencies, are obtained by plugging in Archimedean copulas into each other \citep[pp.~87--89]{tJOE97a}. For instance, in the two-dimensional Archimedean copula
\begin{equation} \nonumber
C_{D_0}(u_1, \bullet)=\psi_{D_0}(\psi_{D_0}^{-1}(u_1) +\psi_{D_0}^{-1}(\bullet)),
\end{equation}
the argument $\bullet$ can be replaced by another Archimedean copula, such as

\begin{equation} \nonumber
\boldsymbol{C_{D_{23}}(u_2, u_3)=\psi_{D_{23}}(\psi_{D_{23}}^{-1}(u_2) +\psi_{D_{23}}^{-1}(u_3))},
\end{equation}
in order to get a copula of the form
\begin{multline} \label{B_0}
C_{D_0}(u_1, \boldsymbol{C_{D_{23}}(u_2, u_3)})= \\ \psi_{D_0}\big(\psi_{D_0}^{-1}(u_1) +\psi_{D_0}^{-1}(\boldsymbol{\psi_{D_{23}}(\psi_{D_{23}}^{-1}(u_2) +\psi_{D_{23}}^{-1}(u_3))})\big).
\end{multline}

This last equation describes a copula where the bivariate marginal distribution of ($U_2$, $U_3$) is not the same as the bivariate marginal distribution of ($U_1$, $U_2$) or ($U_1$, $U_3$), provided the generators $\psi_{D_0}$ and $\psi_{D_{23}}$ are different. If the joint CDF of $(U_1, U_2, U_3)$ was a simple Archimedean copula, all the bivariate marginal distributions would have been identical. This allows one to appreciate how the symmetry inherent in Archimedean copulas can be broken, although some partial symmetry always remains, as the bivariate marginal distributions of $(U_1, U_2)$ and $(U_1, U_3)$ still coincide.

The way Archimedean copulas are nested corresponds to a rooted tree structure, which will be referred to as the \emph{NAC tree structure} or sometimes simply as the \emph{structure} later on. Nested Archimedean copulas, such as the one in (\ref{B_0}), are defined through that rooted tree structure and through a collection of generators, one for each node in the tree that is not a leaf. If the only nodes in the tree are the root and the leaves, then the copula is an Archimedean copula, that is, a nested Archimedean copula with trivial structure and only one generator.

\begin{definition}
\label{def:tree}
Let $D_0$ be a nonempty, finite set with $|D_0| = d$ elements. For concreteness, let $D_0 = \{U_1, \ldots, U_d\}$. Formally, a \emph{rooted tree structure} $\lambda$ on $D_0$ is a collection of nonempty subsets of $D_0$ such that
\begin{enumerate}[(i)]
  \setlength{\itemsep}{1pt}
  \setlength{\parskip}{0pt}
  \setlength{\parsep}{0pt}
\item $D_0 \in \lambda$;
\item $\{ U_j \} \in \lambda$ for every $U_j \in D_0$;
\item if $A, B \in \lambda$, then either $A \subset B$, $B \subset A$, or $A \cap B = \varnothing$.
\end{enumerate}
The elements of $\lambda$ are called the \emph{nodes} of the structure. The element $D_0$ of $\lambda$ is called the \emph{root node}, or \emph{root} in short; the singleton elements $\{ U_j \}$ of $\lambda$ are called the \emph{leaves}. The nodes of $\lambda$ that are not leaves are called the \emph{branching nodes}.
If $A, B \in \lambda$ are such that $A \subset B$, $A \ne B$, and there is no $C \in \lambda$ such that $A \subset C \subset B$ and $C \ne A$ and $C \ne B$, then $A$ is called a \emph{child} of $B$ and conversely $B$ is called the \emph{parent} of $A$. The set of children of $B$ in $\lambda$ is denoted by $\mathcal{C}(B, \lambda)$.
\end{definition}

For instance, the structure $\lambda$ implied by Equation (\ref{B_0}) is
\begin{equation} \nonumber
\bigl\{ \{ U_1, U_2, U_3 \}, \{ U_2, U_3 \}, \{ U_1 \}, \{ U_2 \}, \{ U_3 \} \bigr\},
\end{equation}
and it can be graphically represented as shown in the left-hand panel of Figure~\ref{picture_tree}, where $D_{123}$ is a convenient label for the subset $\{U_1, U_2, U_3 \}$ and $D_{23}$ for the subset $\{ U_2, U_3 \}$.

\begin{figure}[H]
\centering
\begin{tabular}{cc}

\includegraphics[width=0.225\textwidth]{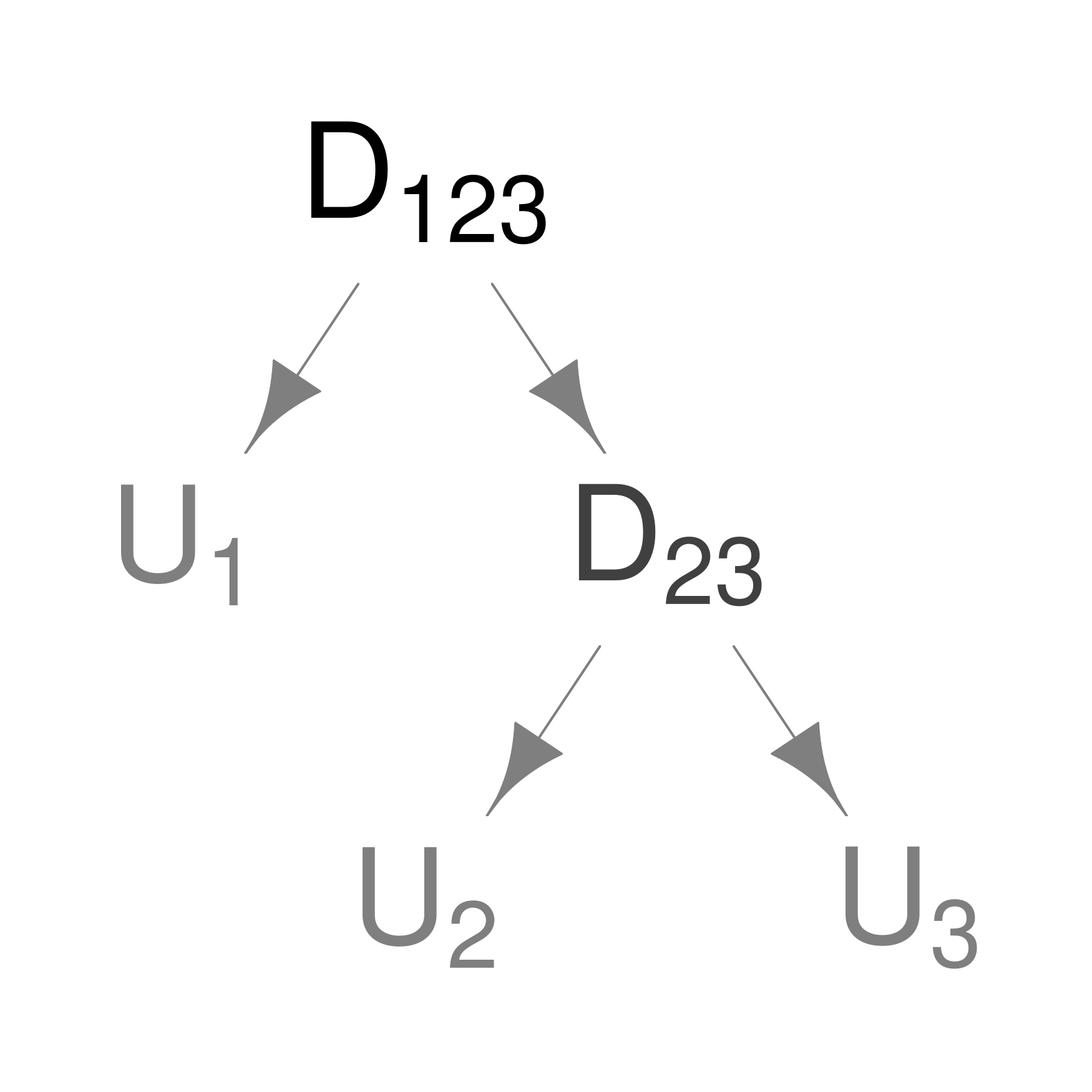}
&
\includegraphics[width=0.225\textwidth]{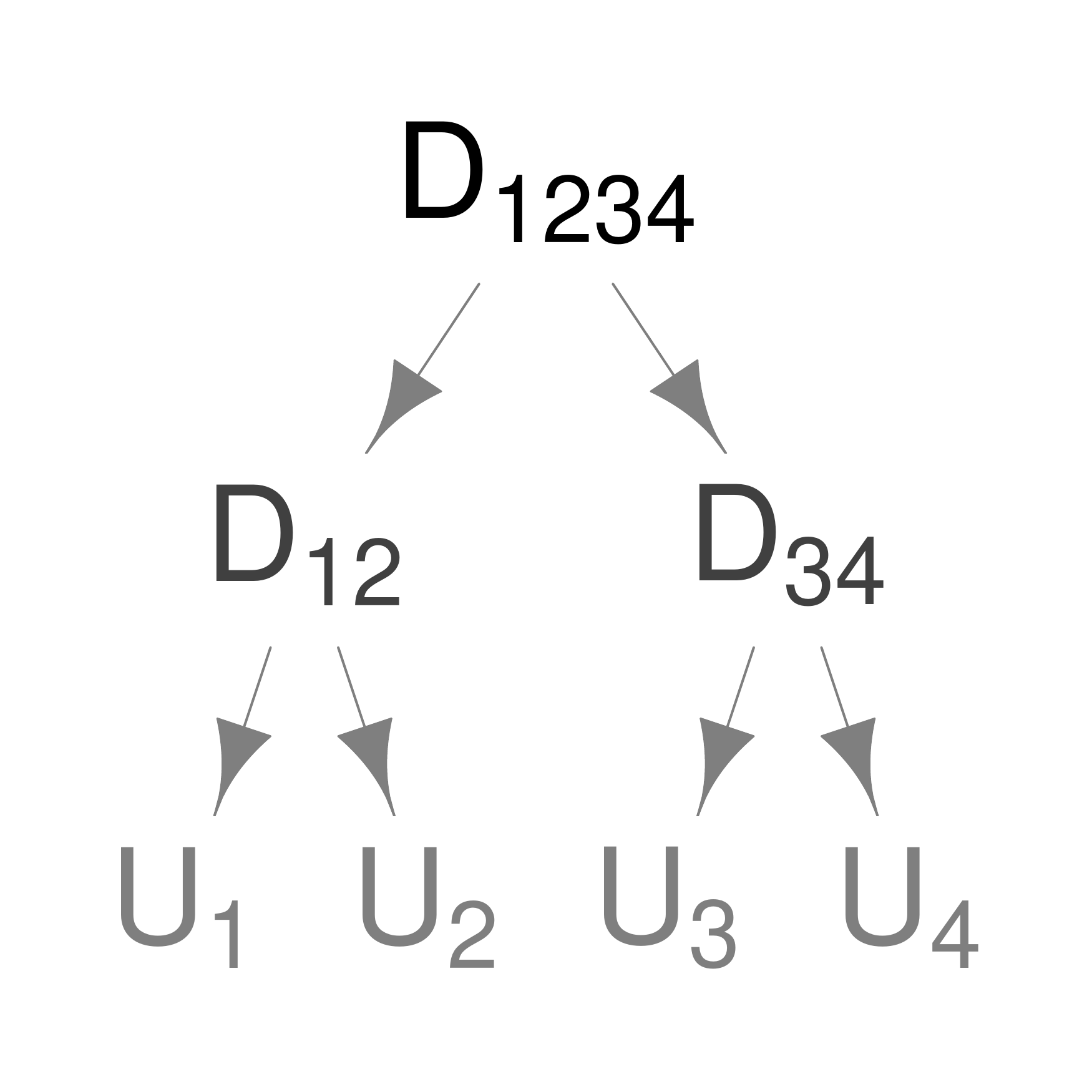}

\end{tabular}
\caption{On the left, the tree structure implied by Equation (\ref{B_0}). To ease the notation, the singletons $\{U_1\}$, $\{U_2\}$ and $\{U_3\}$ are denoted by $U_1, U_2$ and $U_3$. On the right, $D_{12}$ is a convenient label for $\{ U_1, U_2 \}$, as well as $D_{34}$ for the subset $\{ U_3, U_4\}$ and $D_{1234}$ for the set $\{U_1, U_2, U_3, U_4\}$. Again, we ease the notation by writing $U_1, \ldots, U_d$ instead of $\{U_1\}, \ldots, \{U_d\}$ for the singletons.\label{picture_tree}}
\end{figure}

\noindent In the structure on the left in Figure \ref{picture_tree}, $\{U_2\}$ and $\{U_3\}$ are the children of $D_{23}$ while $\{U_1\}$ and $D_{23}$ are the children of $D_{123}$, the root node.

Let $\lambda$ be a rooted tree on $D_0 = \{U_1, \ldots, U_d\}$. Define the related set of indices as $d_{D_0}=\{1, \ldots, d\}$. Suppose that for each $B \in \lambda$ with $|B| \ge 2$ we are given an Archimedean generator $\psi_B$, that is, we are given a generator for each branching node in the structure. Further let the set of indices related to $B$ be denoted as $d_B$.

Next, recursively define the functions $C_B : [0, 1]^{|B|} \to [0, 1]$, with $B \in \lambda$, $|B| \ge 1$, by
\begin{equation}
\label{eq:NAC}
  C_B(u_b : b \in d_B) = 
  \begin{cases}
  u_b 
    & \text{if $B = \{ U_b \}$} \\
  \psi_B 
  \Bigl( 
    \sum_{A \in \mathcal{C}(B, \lambda)} \psi_B^{-1} \bigl( C_A(u_a : a \in d_A) \bigr)
  \Bigr)
    & \text{if $|B| \ge 2$}
  \end{cases}.
\end{equation}

\begin{definition}
\label{def:NAC}
A $d$-variate copula $C_{D_0}$ is a \emph{nested Archimedean copula (NAC)} if it is of the form $C_B$ in~\eqref{eq:NAC}, with $B=D_0$.
\end{definition}

For any $A \subset D_0$ with $|A| \ge 2$, the copula $C_{A}$ on the variables $(u_a : a \in d_A)$ turns out to be a nested Archimedean copula, too. To describe its structure and its generators, we need a few more definitions.

Let $\lambda$ be a NAC structure on $D_0$ and let $A$ be a nonempty subset of $D_0$. The set $A$ need not be a node of $\lambda$. The NAC structure $\lambda$ induces a NAC structure on $A$ by the following operation:
\[
  \lambda \sqcap A= \{ A \cap B : B \in \lambda \} \setminus \{ \varnothing \}.
\]
That is, $\lambda \sqcap A$ is obtained by intersecting every node $B$ of $\lambda$ with $A$. Some of these intersections will be empty, and they are removed. Different nodes $B_1$ and $B_2$ of $\lambda$ may have identical intersections $B_1 \cap A$ and $B_2 \cap A$ with $A$; since $\lambda \sqcap A$ is the collection of all intersections, identical intersections are counted only once. 

It is easy to verify that this construction produces a tree structure on $A$: verification of (i), (ii), and (iii) in Definition~\ref{def:tree} is immediate.

Let $T$ be a subset of $D_0$ containing at least two elements, that is $|T|\geq2$. The set $T$ does not need to be a node of $\lambda$. The \emph{lowest common ancestor} (lca) in $\lambda$ of the elements of $T$ is given by the intersection of all the nodes $B$ in $\lambda$ that contain $T$, that is, 

\begin{equation}
\label{eq:sca_def}
  \lca(T, \lambda) = \bigcap_{ B \in \lambda : T \subseteq B} B,
\end{equation}
and it provides the lowest branching node (read: farthest from the root) through which the elements of $T$ are linked up in $\lambda$. For instance, looking back at Figure~\ref{picture_tree}, one can see that the lowest common ancestor between $U_2$ and $U_3$ is $D_{23}$, while lca($\{ U_1, U_2\}, \lambda)=D_0$ and lca($\{ U_1, U_3\}, \lambda)=D_0$.

Let $C_{D_0}$ be a $d$-variate nested Archimedean copula and let $A$ be a nonempty subset of $D_0$, not necessarily a node in the tree $\lambda$. The marginal copula $C_A$ on the variables in $A$ is a nested Archimedean copula, too. Its NAC structure is given by $\lambda \sqcap A$, and the generator function associated to a branching node $T$ in $\lambda \sqcap A$ is given by $\psi_{\lca(T, \lambda)}$.  

As appealing as it is, Definition \ref{def:NAC} is unfortunately not sufficient to guarantee that $C_{D_0}$ and its margins are copulas. A sufficient but not necessary condition was developed by \citet[pp.~87--89]{tJOE97a} and \cite*{doi:10.1080/00949650701255834}: the derivatives of $\psi_I^{-1} \circ \psi_J$ are required to be completely monotone for every pair of branching nodes $I$ and $J$ in the NAC structure such that $J$ is a child of $I$. As an example, a sufficient condition for $C_{D_0}$ in Equation (\ref{B_0}) to be a proper copula is that the derivatives of $\psi_{D_0}^{-1} \circ \psi_{D_{23}}$ are completely monotone. Although this sufficient nesting condition was originally formulated only in the context of fully nested Archimedean copula structures, that is, structures where each branching node has either two leaves as children, or one leave and another branching node, we assume this sufficient nesting condition to hold for any NAC structure. Also note this sufficient nesting condition can be weakened at least on the lowest nesting level of the structure, as briefly discussed in \cite*{hofert2012stochastic}.

The sufficient nesting condition is often easily verified if all generators appearing in the nested structure come from the same parametric family. For each family of Table \ref{archi_families}, two generators $\psi_I$ and $\psi_J$ of the same family with corresponding parameters $\theta_I$ and $\theta_J$ will fulfill the sufficient nesting condition if $\theta_I \leq \theta_J$, assuming $J$ is the child of $I$. Verifying the sufficient nesting condition if $\psi_I$ and $\psi_J$ do not belong to the same Archimedean family is usually harder, see for instance \cite*{hofert2010sampling}.

\section{Identifiability \label{identifiability_sec}}

Recall that a parameter $\theta$ (possibly infinite-dimensional) in a statistical model $(P_\theta : \theta \in \Theta)$, with $P_\theta$ a probability measure on a fixed space, is identifiable if $\theta_1 \ne \theta_2$ implies that $P_{\theta_1} \ne P_{\theta_2}$, that is, different parameters yield different distributions of the observable. For $d$-variate nested Archimedean copulas, the parameter $\theta$ consists of the pair 
\[ 
  \bigl( \lambda, \{\psi_B: B \in \lambda, |B| \ge 2\} \bigr).
\]

In this parametrization, the parameter $\theta$ is not identifiable, since replacing a generator function $\psi_B(x)$ by the function $\psi_B(ax)$, with $0 < a < \infty$, yields the same copula; that is, the generator functions are identifiable up to scaling only. This issue can be solved easily in different ways, for instance by requiring that $\psi_B(1) = 1/2$. This problem has however no impact on the structure $\lambda$ itself and is therefore of little interest for this paper.

A more fundamental identifiability issue arises if some generator functions are the same. Consider for instance the tree $\lambda$ implied by Equation (\ref{B_0}), shown on the left in Figure \ref{picture_tree}.
If the generators $\psi_{D_0}$ and $\psi_{D_{23}}$ are the same, say $\boldsymbol{\psi}$, then the nested Archimedean copula with parameter $(\lambda; \psi_{D_0}, \psi_{D_{23}}$) is

\begin{align*}
C_{D_0}(u_1, C_{D_{23}}(u_2, u_3))
  &=\psi_{D_0}\big(\psi_{D_0}^{-1}(u_1) +\psi_{D_0}^{-1}(\psi_{D_{23}}(\psi_{D_{23}}^{-1}(u_2) +\psi_{D_{23}}^{-1}(u_3)))\big) \\
  &= \boldsymbol{\psi} \bigl( \boldsymbol{\psi}^{-1}(u_1) + \boldsymbol{\psi}^{-1}(u_2) + \boldsymbol{\psi}^{-1}(u_3) \bigr),
\end{align*}
and actually describes an exchangeable Archimedean copula with generator $\boldsymbol{\psi}$, that is, a nested Archimedean copula with trivial tree structure and single generator $\boldsymbol{\psi}$. 

To ensure identifiability of the structure, we must require that for any two nodes $A$ and $B$ such that $A \subset B$ and $A \ne B$, meaning $A$ is a descendant of $B$, the bivariate Archimedean copulas generated by the generator functions $\psi_A$ and $\psi_B$ are different, prohibiting the tree structure to collapse at some level. If this condition holds, then the structure $\lambda$ can be identified. This weak restriction on the generators will be assumed to hold throughout this paper. 

Note that some generator functions can still be identical. Consider for instance the structure on the right in Figure \ref{picture_tree}. The generators associated to the nodes $D_{12}$ and $D_{34}$ can be identical, without simplification of the tree being possible.

Also note the implication of this identifiability condition on the sufficient nesting condition if all generators appearing in the nested structure come from the same parametric family. For each family of Table \ref{archi_families}, two generators $\psi_I$ and $\psi_J$ of the same family with corresponding parameters $\theta_I$ and $\theta_J$ will fulfill the sufficient nesting condition \emph{and} the identifiability condition if $\theta_I$ is \emph{strictly less} than $\theta_J$, assuming $J$ is a child of $I$.

\section{Nonparametric estimation of a trivariate NAC structure \label{test}}

Let $(X_1, X_2, X_3)$ be a vector of continuous random variables such that the joint distribution of $(U_1, U_2, U_3)=(F_{X_1}(X_1), F_{X_2}(X_2), F_{X_3}(X_3))$ is a nested Archimedean copula. We are interested in estimating the NAC structure based on $n$ observations $(x_{l1}, x_{l2}, x_{l3})$ from $(X_1, X_2, X_3)$, $l=1, \ldots, n$.

There are only four possible structures fulfilling Definition \ref{def:tree} for the trivariate case:
 {\small
 \[
   \begin{array}{rcl}
     \bigl\{ \{ U_1, U_2, U_3 \}, \{ U_1 \}, \{ U_2 \}, \{ U_3 \} \bigr\} &=& \Lambda_{123}; \\[1ex] 
     \bigl\{ \{ U_1, U_2, U_3 \}, \{ U_2, U_3 \}, \{ U_1 \}, \{ U_2 \}, \{ U_3 \} \bigr\} &=& \lambda_{23};\\[1ex]
     \bigl\{ \{ U_1, U_2, U_3 \}, \{ U_1, U_2 \}, \{ U_1 \}, \{ U_2 \}, \{ U_3 \} \bigr\}&=& \lambda_{12}; \\[1ex]
     \bigl\{ \{ U_1, U_2, U_3 \}, \{ U_1, U_3 \}, \{ U_1 \}, \{ U_2 \}, \{ U_3 \} \bigr\} &=& \lambda_{13}.
   \end{array}
 \]
 }In the trivial structure (tree $\Lambda_{123}$), all bivariate marginal distributions of the nested Archimedean copula are the same, while in structures $\lambda_{23}$, $\lambda_{12}$ and $\lambda_{13}$, two bivariate marginal distributions are the same and one is different. Moreover, if the bivariate marginal distributions are not all the same, being able to determine the one that is different from the two others is enough to select the proper nested Archimedean copula structure $\lambda_{23}$, $\lambda_{12}$ or $\lambda_{13}$.

It is known from \cite{genest1993} that the Kendall distribution of a pair of variables $(X_j, X_k)$ fully determines the copula of that pair provided the copula is Archimedean. Thus, rather than working directly with bivariate distributions, let us work with the related Kendall distributions which are univariate and therefore easier to handle. The Kendall distribution of the pair $(X_j$, $X_k)$ is defined as the distribution of the variable

\begin{equation}
W_{jk}=C_{jk}(U_j, U_k)=H_{jk}(X_j, X_k), \nonumber
\end{equation}

\noindent where $C_{jk}(u_j, u_k)=P(U_j \leq u_j, U_k \leq u_k)$ is the joint CDF of $(U_j, U_k)$, and where $H_{jk}(x_j, x_k)=P(X_j \leq x_j, X_k \leq x_k)$ is the joint CDF of $(X_j, X_k)$. The map defined for all $w \in [0,1]$ by

\begin{equation}
K_{jk}(w)=P(W_{jk} \leq w), \nonumber
\end{equation}

\noindent is the Kendall distribution function (\citeauthor{barbe1996kendall} \citeyear{barbe1996kendall}; \citeauthor{nelsen2003kendall} \citeyear{nelsen2003kendall}; \citeauthor{genest2001multivariate} \citeyear{genest2001multivariate}).

The Kendall distribution function of a pair of variables $(X_j, X_k)$ can be estimated \citep*{genest2011inference} by first computing the pseudo-observations $w_{1,jk}, \ldots, w_{n,jk}$ and then the empirical distribution function of these pseudo-observations:

\begin{equation*}
w_{m,jk}=\frac{1}{n+1}\sum_{l=1}^n 1(x_{lj}<x_{mj}, x_{lk}<x_{mk});
\label{pseudo_obs_kendall}
\end{equation*}



\[
K_{n,jk}(w)=\frac{1}{n} \sum_{m=1}^n 1(w_{m,jk} \leq w), \text{ with } 0<w<1.
\]



Since there are three possible pairs in our case, namely $(X_1, X_2), (X_1, X_3)$ and $(X_2, X_3)$, three empirical Kendall distribution functions need to be estimated. The distance between the empirical Kendall distribution functions of $(X_i, X_j)$ and $(X_i, X_k)$ is defined as

\begin{equation}
\int_0^1 |K_{n,ij}(x)-K_{n,ik}(x)|\,\,dx= \frac{1}{n}\sum_{m=1}^n |w_{(m), ij}-w_{(m), ik}|=\delta_{ij,ik}, \nonumber
\end{equation}
where $w_{(1), ij}, \ldots, w_{(n), ij}$ are the ordered pseudo-observations related to the variables $(X_i, X_j)$ and $w_{(1), ik}, \ldots, w_{(n),ik}$ are those related to the variables $(X_i, X_k)$.

Typically, a trivial structure will result in three distances that are all about the same, while trees such as $\lambda_{12}$, $\lambda_{13}$ or $\lambda_{23}$ will result in one small distance relative to two other distances that are bigger and about the same. Thus for any three variables $(X_i, X_j, X_k)$, if, for instance, $\delta_{ij,ik}$ is the minimum among the three distances, it seems reasonable to assume that the tree spanned on $(X_i, X_j, X_k)$ is either the trivial structure or the structure $\lambda_{jk}$ where $(X_i, X_j)$ and $(X_i, X_k)$ have the same Kendall distribution.

The problem of determining the structure of $(X_1, X_2, X_3)$ can thus be rewritten as an hypothesis test:

\begin{table}[h]
\centering
\small
\begin{tabular}{ll}
$H_0:$& the true structure is the trivial structure. \\
$H_1:$& the true structure is structure $\lambda_{12}$ or $\lambda_{13}$ or $\lambda_{23}$, depending on which\\ 
 &pair of Kendall functions were the closest.
\end{tabular}
\end{table}
As a test statistic, the absolute difference between the minimum distance and the average of the two remaining distances is used. The null hypothesis is rejected when the test statistic is observed in the upper tail of its $H_0$ distribution.

As the $H_0$ distribution of the test statistic is unknown, we rely on the bootstrap to calculate p-values. Under $H_0$ the original sample is assumed to come from an unknown trivariate Archimedean copula. Using the work of \cite{genest2011inference}, it is possible to estimate that Archimedean copula nonparametrically and to resample from that estimated Archimedean copula. For each new sample one obtains the three empirical Kendall distributions, the three distances, and the related test statistic. The p-value of the observed test statistic from the original sample is then estimated by the proportion of test statistics obtained from the new samples that are greater than or equal to the value of the observed test statistic from the original sample. Should this estimated p-value be lower than a significance level $\alpha$, for instance 10\%, the null hypothesis is rejected.

Since the estimator for the Kendall distribution depends on the data only through the ranks and since our test statistic only depends on this estimator, our NAC structure estimator is rank-based, too.

There are two key points in the test presented above:

\begin{itemize}
  \setlength{\itemsep}{1pt}
  \setlength{\parskip}{0pt}
  \setlength{\parsep}{0pt}
\item First, determine which should be the alternative hypothesis. Should it be structure $\lambda_{12}$, $\lambda_{13}$ or $\lambda_{23}$?
\item Second, choose between a trivial structure (= $H_0$) and $H_1$.

\end{itemize}

\noindent Possible errors are:

\begin{itemize}
  \setlength{\itemsep}{1pt}
  \setlength{\parskip}{0pt}
  \setlength{\parsep}{0pt}
\item If the true structure is the trivial structure, rejecting it and therefore committing a type I error;
\item If the true structure is structure $\lambda_{12}$, $\lambda_{13}$ or $\lambda_{23}$, failing to reject $H_0$ (type II error);
\item If the true structure is for instance structure $\lambda_{12}$, getting a wrong $H_1$ and then picking $H_1$ (we will call this a type III error).

\end{itemize}


The main difficulty with the test developed in this section is encountered when the true structure is the trivial trivariate structure, that is, the structure one gets when the nested Archimedean copula is actually an exchangeable Archimedean copula. Indeed if the probability of committing a type I error is fixed to $\alpha=0.10$, the trivial structure will be rejected 10\% of the time regardless of the input sample size $n$. Our estimator is therefore not a consistent estimator for the trivial trivariate structure, unless we let $\alpha$ tend to 0 as $n$ increases, so that type I errors are asymptotically impossible. Practically speaking however, this rule has little significance as it does not help to select $\alpha$ given a value of $n$. In the simulation section of this paper, $\alpha$ will be fixed to $10\%$ for all $n$, yielding satisfactory performance.

\section{Recovering a target structure from trivariate structures \label{sufficiency_triple}}
In Section \ref{test}, we showed how to infer the tree structure for three variables at a time. Next, we need to assemble these trivariate structures into a single $d$-variate structure. For this to be possible, we need to ensure that the full tree is indeed determined by the tree structures it induces on the collection of subsets of three variables.

Let $\lambda$ be a NAC structure on $D_0=\{ U_1, \ldots , U_d \}$, $d\geq 3$. Let $B$ be a branching node of $\lambda$. The set of all children of $B$ forms a partition of $B$, that is, taking the union of all children of a branching node $B$ allows to reconstruct that branching node. As a consequence, every branching node has at least two children.

Since the children of a branching node $B$ form a partition of $B$ and since each branching node has at least two children, it follows that each branching node can be reconstructed from the pairs of which it is the lowest common ancestor, that is, for every branching node $B$, we have
\begin{align}
  B = \bigcup \bigl\{ \{ U_i, U_j \} \subset D_0 : U_i \ne U_j, \, \lca(\{U_i , U_j \}, \lambda) = B \bigl\}. 
  \label{reconstruct}
\end{align}

The relation ``\ldots has the same lowest common ancestor as \ldots'' is an equivalence relation on the set of pairs $\{U_i, U_j \}$ of $D_0$. This relation induces a partition of the set of pairs into equivalence classes: two pairs $\{U_i, U_j \}$ and $\{ U_k, U_l \}$ belong to the same equivalence class if and only if they have the same lowest common ancestor in $\lambda$.

By Equation (\ref{reconstruct}), the nested Archimedean copula structure $\lambda$ can be reconstructed from the equivalence relation it induces on the set of pairs: every equivalence class of pairs corresponds to a branching node, the branching node being given by the union of the pairs in that equivalence class. Put differently, the union of all pairs within an equivalence class yields the branching node that is the lowest common ancestor for each pair in that equivalence class. Hence, every NAC structure $\lambda$ on $D_0$ can be represented as a partition on the set of pairs of $D_0$.

Let $d \geq 4$. Suppose that for any set $K_{ijk}=\{ U_i, U_j, U_k \}$ with distinct $i, j, k \in \{1, \ldots, d\}$, the tree spanned on $\{ U_i, U_j, U_k \}$, $\lambda \sqcap K_{ijk}$, is known. Define $^{3}(\lambda)$ as the set of these $\bigl(\begin{smallmatrix}
d \\ 3
\end{smallmatrix} \bigr)$ trees.

In Proposition \ref{main:prop}, it is shown that the nested Archimedean copula structure $\lambda$ can be recovered from $^{3}(\lambda)$. Lemmas \ref{lem:scaspaneq} and \ref{lem:sca} contain some auxiliary results, with proofs in the Appendix.








\begin{lemma}
\label{lem:scaspaneq}
Let $\lambda$ be a tree on $D_0$. For any nonempty subsets $T_1, T_2, C$ of $D_0$ such that $T_1 \cup T_2 \subset C$, we have
\begin{align*}
  &\lca( T_1, \lambda ) = \lca( T_2, \lambda ) \iff \lca( T_1, \lambda \sqcap C ) = \lca( T_2, \lambda \sqcap C ).
\end{align*}
\end{lemma}

Essentially this lemma states that if two subsets of $D_0$ have the same lowest common ancestor in $\lambda$, then they also have the same lowest common ancestor in any subtree of $\lambda$, provided the two subsets are included in the set of leaves of that subtree. 
For instance, consider the structure in the left-hand panel of Figure \ref{sevenvariate_sim}, where $D_0=\{ U_1, \ldots , U_7 \}$. The set $\{ U_5, U_6 , U_7 \}$ has the same lowest common ancestor as the set  $\{ U_5, U_7 \}$ in $\lambda$, this lowest common ancestor being the node $D_{567}$. This holds true even if you consider only the subtree spanned on $\{U_4, U_5, U_6 , U_7 \}$, that is, even if you only consider the right branch of the structure in Figure \ref{sevenvariate_sim}.



\begin{lemma}
\label{lem:sca}
Let $\lambda$ be a tree on $D_0$ and let $A \in \lambda$. Let $B$ be a nonempty subset of $D_0$ with a least two elements. The lowest common ancestor of $B$ is equal to $A$ if and only if $B \subset A$ and there exist distinct children $B_1$ and $B_2$ of $A$ such that $B \cap B_1 \ne \varnothing$ and $B \cap B_2 \ne \varnothing$.
\end{lemma}

The meaning of this lemma is less straightforward. It states that if $B$ is a subset of $A$, $A$ being a node of $\lambda$, the only way $A$ is going to be the lowest common ancestor of $B$ is if $B$ has a nonempty intersection with two distinct children of $A$. 
Consider Figure \ref{sevenvariate_sim} again. If $A=\{U_4, U_5, U_6 , U_7 \}$, then the lca of $B=\{U_4, U_5\}$ is $A$ because in that case, $B$ has a nonempty intersection with the only two children of $A$, these two children being $\{U_4\}$ and $D_{567}=\{U_5, U_6 , U_7 \}$.



\begin{proposition}
\label{main:prop}
The NAC structure $\lambda$ can be recovered from the set $^{3}(\lambda)$, that is, it is possible to retrieve the partition of the set of pairs $\{U_i, U_j\}$ of $D_0$ into equivalence classes from the set $^{3}(\lambda)$.
\end{proposition}

\begin{proof}
Let first $\{ U_i, U_j \}$ and $\{ U_i, U_k \}$ be two pairs with exactly one element, $U_i$, in common. To see whether they have the same lowest common ancestor in $\lambda$, it is sufficient to consider the tree induced by $\lambda$ on the set $\{ U_i, U_j, U_k \}$: it is known from Lemma \ref{lem:scaspaneq} that the pairs $\{ U_i, U_j \}$ and $\{ U_i, U_k \}$ have the same lowest common ancestor in $\lambda$ if and only if they have the same lowest common ancestor in $\lambda \sqcap \{ U_i, U_j, U_k \}$.

On the other hand, if two pairs are disjoint, there exists no set with only three elements containing both pairs. Still, considering $\{K_{ijk}\}$ with distinct $i, j, k \in \{1, \ldots, d\}$ turns out to be sufficient to verify the equivalence of the two disjoint pairs: the two pairs can only be equivalent if there is a third pair equivalent to both of them and having a nonempty intersection with each of them.

Indeed suppose first there exists a pair $\{ U_i, U_j\}$ having the same lowest common ancestor as the pair $\{ U_i, U_k\}$. Also suppose $\{ U_i, U_k\}$ has the same lowest common ancestor as $\{ U_k, U_l\}$. Then by transitivity $\{ U_i, U_j\}$ has the same lowest common ancestor as $\{ U_k, U_l\}$. 

Conversely, suppose that $\{ U_k, U_l\}$ and $\{ U_i, U_j\}$ have the same lowest common ancestor, $A$. Recall Lemma \ref{lem:sca}. Let $B_i, B_j, B_k, B_l$ be the children of $A$, to which $U_i, U_j, U_k, U_l$ belong, respectively. Since the lca of $U_i$ and $U_j$ is $A$, the pair $\{U_i, U_j\}$ must have a non-empty intersection with two different children of $A$ (Lemma \ref{lem:sca}). Hence, $B_i$ and $B_j$ must be disjoint, $B_i \cap B_j = \varnothing$. Similarly $B_k \cap B_l = \varnothing$. Then $B_k$ and $B_l$ cannot both be equal to $B_i$.

\begin{itemize}
  \setlength{\itemsep}{1pt}
  \setlength{\parskip}{0pt}
  \setlength{\parsep}{0pt}
\item If $B_k$ is different from $B_i$, then $U_i$ and $U_k$ belong to two different children of $A$, and the lowest common ancestor of $\{ U_i, U_k \}$ is $A$, too;
\item If $B_l$ is different from $B_i$, then, similarly, the lowest common ancestor of $\{ U_i, U_l \}$ is $A$, too.
\end{itemize}
In both cases, we have found a pair that is equivalent to $\{ U_i, U_j\}$ and $\{ U_k, U_l\}$ and that has a nonempty intersection with each of them.
\end{proof}

Given a nested Archimedean copula structure $\lambda$, it is thus always possible to break it down into a set of trivariate structures, one trivariate structure for each combination of the elements of $D_0$, taken three at the time without repetition. Proposition \ref{main:prop} states that this set of trivariate structures is sufficient to recover $\lambda$.

The idea to decompose a structure into smaller pieces that uniquely determine the structure is not new however. \cite*{ng1996} show that a given structure can be broken down into a set of triples and fans. A \emph{triple} is a tree with three leaves and two internal vertices. A \emph{fan} is a tree with only one internal vertex and at least three leaves (that is, a fan is a trivial tree with at least three leaves). A fan with $d$ leaves is called a $d$-fan.

Hereafter is a practical example on how to retrieve $\lambda$ from $^{3}(\lambda)$ when $d=4$. Suppose indeed the $\bigl(\begin{smallmatrix}
4 \\ 3
\end{smallmatrix} \bigr)=4$ elements of $^{3}(\lambda)$ are as shown in Figure \ref{recover4}.

\begin{figure}[H]
\centering
\begin{tabular}{cccc}

\includegraphics[width=0.215\textwidth]{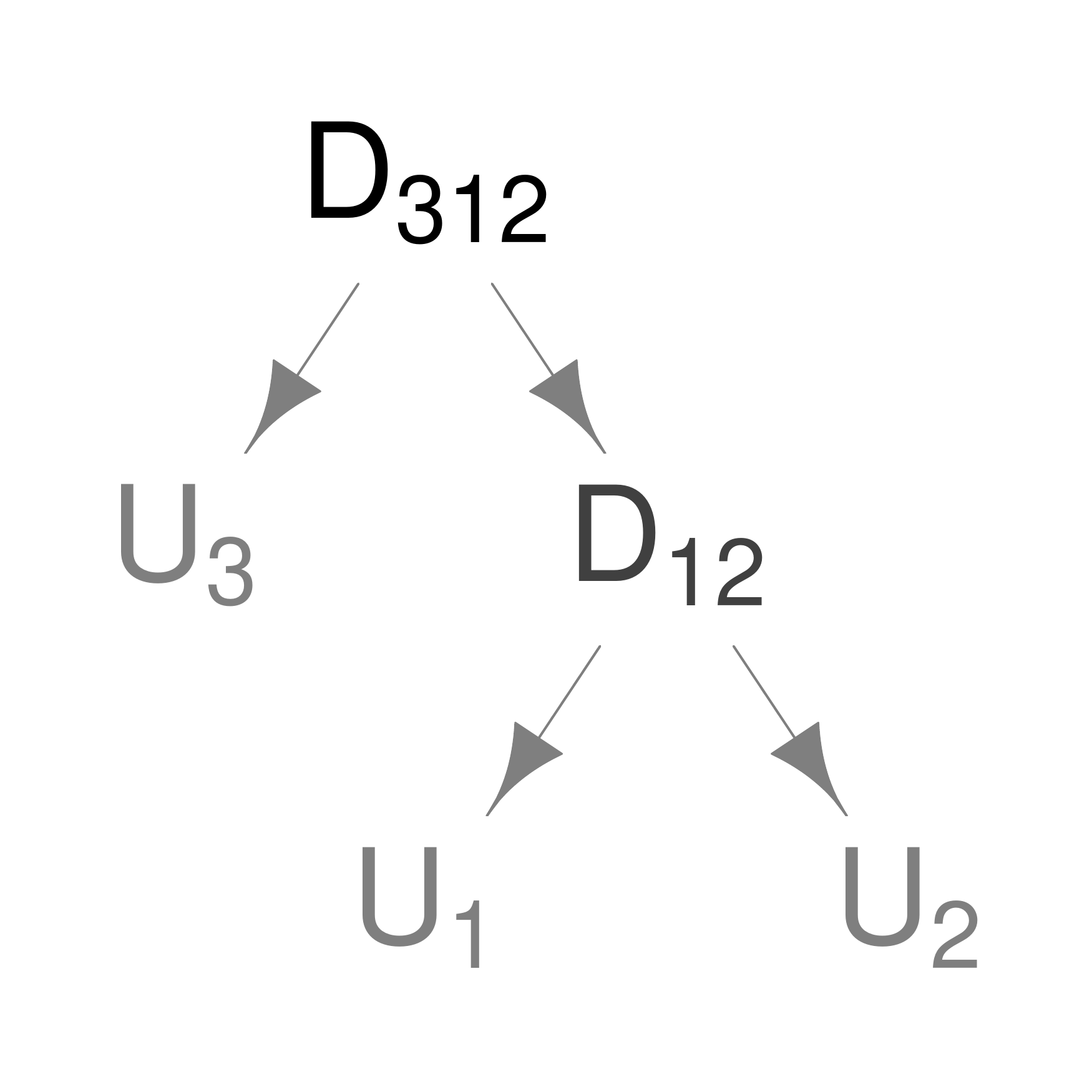}
&
\includegraphics[width=0.215\textwidth]{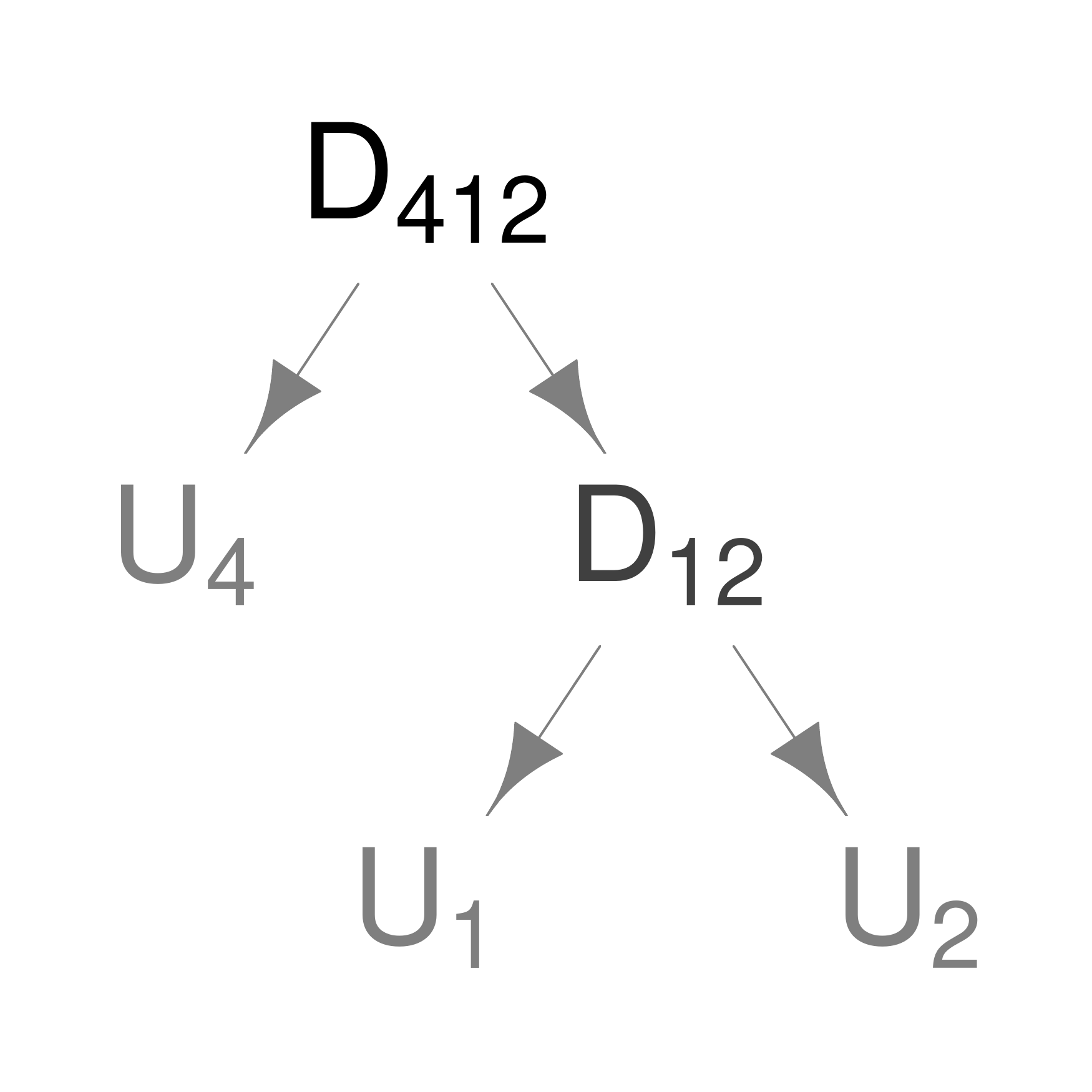}
&
\includegraphics[width=0.215\textwidth]{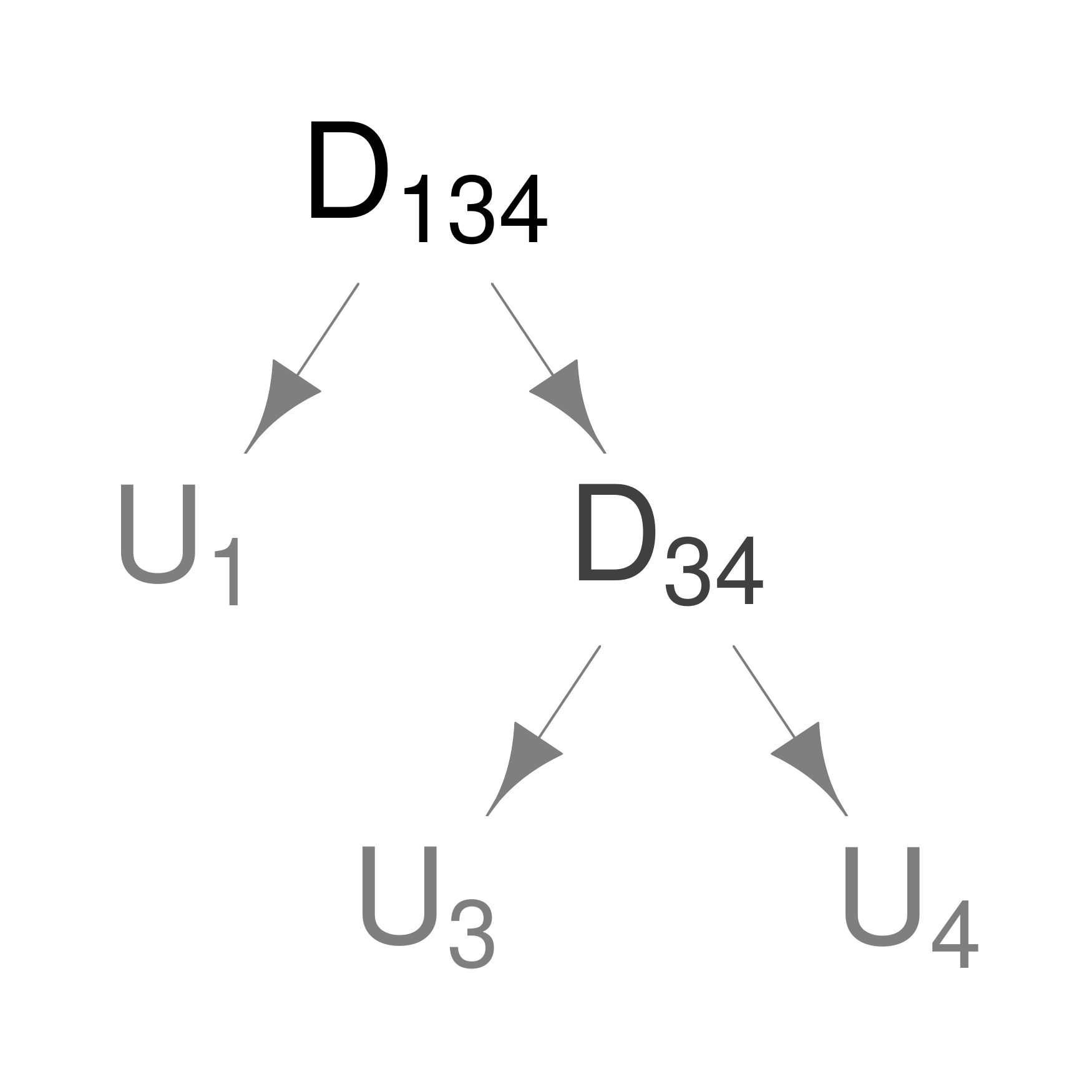}
&
\includegraphics[width=0.215\textwidth]{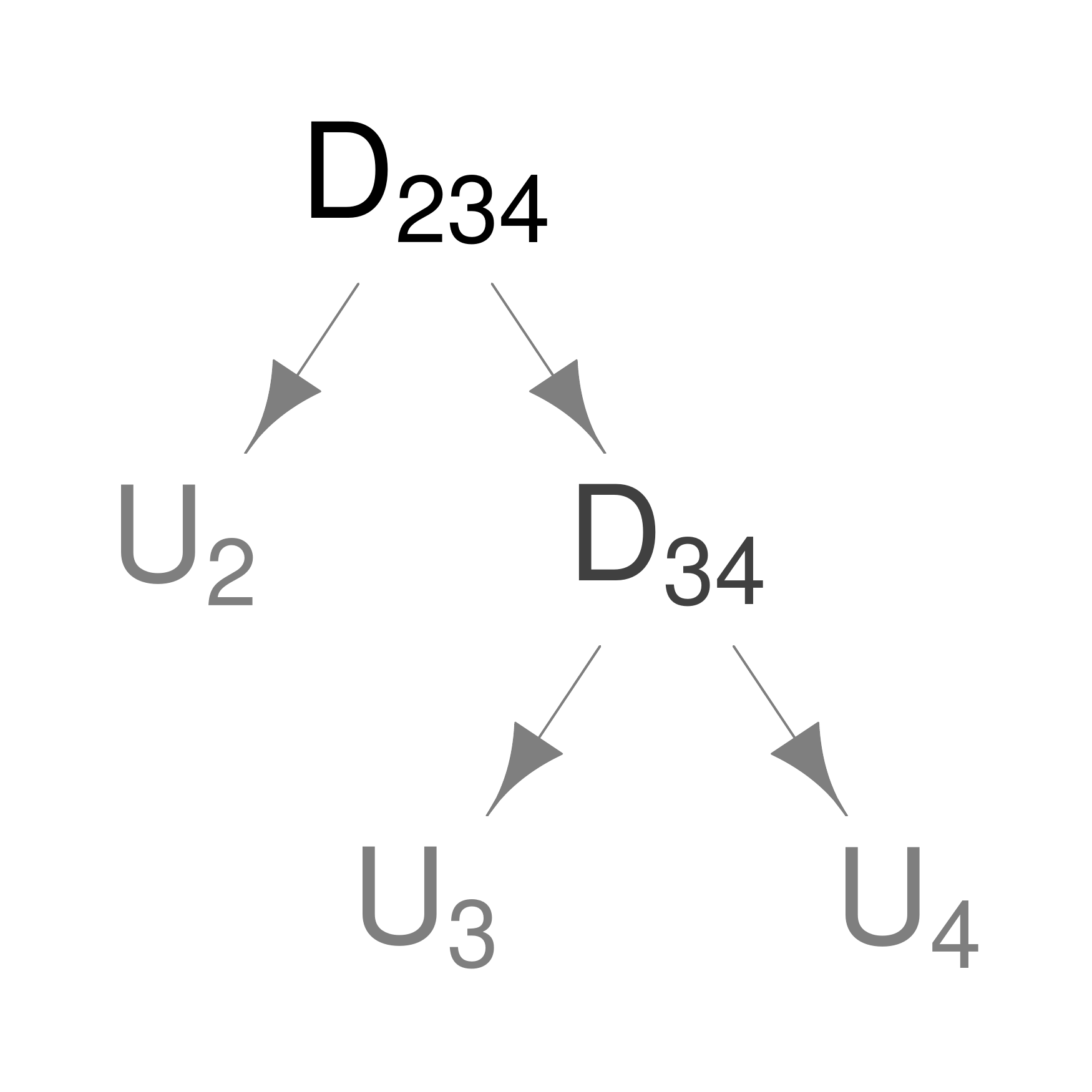}

\end{tabular}
\caption{A set of trivariate structures that uniquely determines the four-variate structure in Figure \ref{picture_tree}. Note the labels of the internal nodes are irrelevant. The lowest common ancestor of $U_1$ and $U_2$ could have been labeled $P$ instead of $D_{12}$ in the first structure and $G$ in the second structure. \label{recover4}}
\end{figure}

From Figure \ref{recover4}, we get that
\begin{itemize}
  \setlength{\itemsep}{1pt}
  \setlength{\parskip}{0pt}
  \setlength{\parsep}{0pt}
\item The lowest common ancestors of the pair $\{U_1, U_2\}$ are $\{D_{12}, D_{12}\}$;
\item The lowest common ancestors of the pair $\{U_1, U_3\}$ are $\{D_{321}, D_{134}\}$;
\item The lowest common ancestors of the pair $\{U_1, U_4\}$ are $\{D_{412}, D_{134}\}$;
\item The lowest common ancestors of the pair $\{U_2, U_3\}$ are $\{D_{312}, D_{234}\}$;
\item The lowest common ancestors of the pair $\{U_2, U_4\}$ are $\{D_{412}, D_{234}\}$;
\item The lowest common ancestors of the pair $\{U_3, U_4\}$ are $\{D_{34}, D_{34}\}$.
\end{itemize}

It appears therefore that $\{U_1, U_3\}, \{U_1, U_4\}, \{U_2, U_3\}$ and $\{U_2, U_4\}$ belong to the same equivalence class, while $\{U_1, U_2\}$ is by itself, as well as $\{U_3, U_4\}$. The branching nodes of $\lambda$ in this case are therefore $\{U_1, U_2, U_3, U_4\},\allowbreak \{U_1, U_2\}$ and $\{U_3, U_4\}$. The rooted tree structure $\lambda$ is thus as shown in Figure \ref{picture_tree}.

The general procedure for any $d \geq 4$ is as shown in Algorithm \ref{algo_retrieve}.


\begin{algorithm}[H]
\caption{How to retrieve $\lambda$ from $^{3}(\lambda)$ \label{algo_retrieve}}
\begin{algorithmic}
\FORALL{pairs $\{U_i, U_j\}$ such that $i < j$}
  \STATE{Get from $^{3}(\lambda)$ the set of lowest common ancestors. There should be $d-2$ lowest common ancestors available for each pair;}
\ENDFOR
\FORALL{pairs $\{U_i, U_j\}$ such that $i < j$}
  \STATE{Intersect the set of lowest common ancestors of the working pair with the other sets (one set for each other pair of variables). Any nonempty intersection means the two pairs are related, that is, belong to the same equivalence class. This also allows to determine the number of equivalence classes;}
\ENDFOR
\FORALL{equivalence classes}
  \STATE{Take the union of all pairs within each equivalence class to get the branching nodes of the structure. There are as many branching nodes as there are equivalence classes;}
\ENDFOR
  \STATE{Add the leaves to the branching nodes to get $\lambda$.}
\end{algorithmic}
\end{algorithm}

\section{Reconstruction of a NAC structure based on a set of estimated trivariate structures \label{problem}}

Let $\lambda$ be a NAC structure on a finite set $D_0=\{U_1, \ldots, U_d\}, d \geq 4$. It is known from Section \ref{sufficiency_triple} that if the tree spanned on any three distinct elements of $D_0$ is known (that is, each element of $^{3}(\lambda)$ is known), then $\lambda$ can be uniquely recovered, for instance using the algorithm at the end of the same section.

Our suggestion for \emph{estimating} $\lambda$ is therefore to estimate, using the procedure developed in Section \ref{test}, each element of $^{3}(\lambda)$, thus effectively getting $\widehat{^{3}(\lambda)}$ which can then be used to build $\hat{\lambda}$.

However if each element of $^{3}(\lambda)$ is estimated, the problem of reconstructing a tree from that set of estimated trivariate trees is a bit different than what was considered in Section \ref{sufficiency_triple}. Indeed it is not guaranteed that a proper nested Archimedean copula structure can be recovered from a given set of estimated trivariate structures, that is, $\widehat{^{3}(\lambda)}$ does not necessarily lead to a proper tree $\hat{\lambda}$. When $\hat{\lambda}$ retrieved from $\widehat{^{3}(\lambda)}$ is not a proper nested Archimedean copula structure, meaning it does not fulfill Definition \ref{def:tree}, we call $\widehat{^{3}(\lambda)}$ a \emph{faulty} set.

With a value of $\alpha$ equal to 0.00 for all tests required to estimate $^{3}(\lambda)$, we fail to reject the null hypothesis everywhere and we therefore get a set of estimated trivariate structures each describing a 3-fan. Such a set is never a faulty set, and $\hat{\lambda}$, the estimated NAC structure retrieved from it, will always be a trivial structure of dimension $d$, a $d$-fan. Of course if the true structure is not a $d$-fan, a value of $\alpha$ equal to 0.00 means you are sure to commit type II errors.

With a value of $\alpha$ equal to 1.00 for all tests, all null hypotheses are rejected and we end up with a set where each estimated trivariate structure describes a triple. Such a set can be a faulty set and usually is.

Assuming the copula of the vector $(X_1, \ldots, X_d)$ is a nested Archimedean copula, a faulty set of estimated trivariate structures means at least one error (type~I, type~II or type~III) has been committed. Notice the converse is not true: even when at least one type~I, type~II or type~III error has been committed, $\widehat{^{3}(\lambda)}$ might lead to a structure $\hat{\lambda}$ meeting Definition~\ref{def:tree}, however not equal to $\lambda$, the target structure. Suppose indeed the target structure is the structure on the right in Figure \ref{picture_tree}. We know from Section \ref{sufficiency_triple} that this structure is uniquely determined by the four triples shown in Figure \ref{recover4}. Suppose however that the first two triples in Figure \ref{recover4} are replaced by two 3-fans, that is, two type II errors have been committed during the estimation process. Yet, this leads to a proper four-variate structure, shown on the left in Figure \ref{four-variate_sim}, however unequal to the target structure.

A faulty set is essentially a red flag that should be viewed as an opportunity for correction. However what kind of corrective measure should be applied to such a set is not straightforward. As done in the simulation study, we simply suggest to decrease the value of $\alpha$ for all tests until the resulting set of estimated trivariate structures is not a faulty set anymore. At worst, $\alpha$ is to be decreased down to 0.00, we end up with a set of 3-fans, and $\hat{\lambda}$ is then a $d$-fan. This last strategy is certainly not the best one can imagine, but is a very convenient one to apply and ensures that $\hat{\lambda}$ will always be a proper tree.

Since the estimator developed in Section \ref{test} is unable to consistently estimate a 3-fan if we keep the same value of $\alpha>0$ for all $n$, it also means we will be unable to consistently estimate any $\lambda$ that has at least one trivial trivariate component, see for instance the simulation results from Figure \ref{four-variate_sim}.

\section{Simulation study \label{sim_sec}}

\subsection{Testing our method with samples from a 3-fan, a 5-fan and a four-variate structure containing two 3-fans.}

Let $(X_1, X_2, X_3)$ be a vector of random variables, the copula of which is Archimedean. We generate 500 samples of size $n$ from $(X_1, X_2, X_3)$ with the help of the \textsf{R} package \textbf{nacopula} \citep{Hofert:Maechler:2010:JSSOBK:v39i09}. \emph{Please note that this package has since been merged with the \textbf{copula} package}. With $\alpha$ arbitrarily set to 0.10, how many times among the 500 samples are we able to retrieve the 3-fan? Figure \ref{AC_sim} shows the percentage of correct estimates for various values of $n$, various generator families and two different values of the related parameter $\theta$, expressed as Kendall's $\tau$ coefficient for convenience according to Table \ref{archi_families}.

\begin{figure}[H]
\centering
\begin{tabular}{cc}

\includegraphics[width=0.45\textwidth]{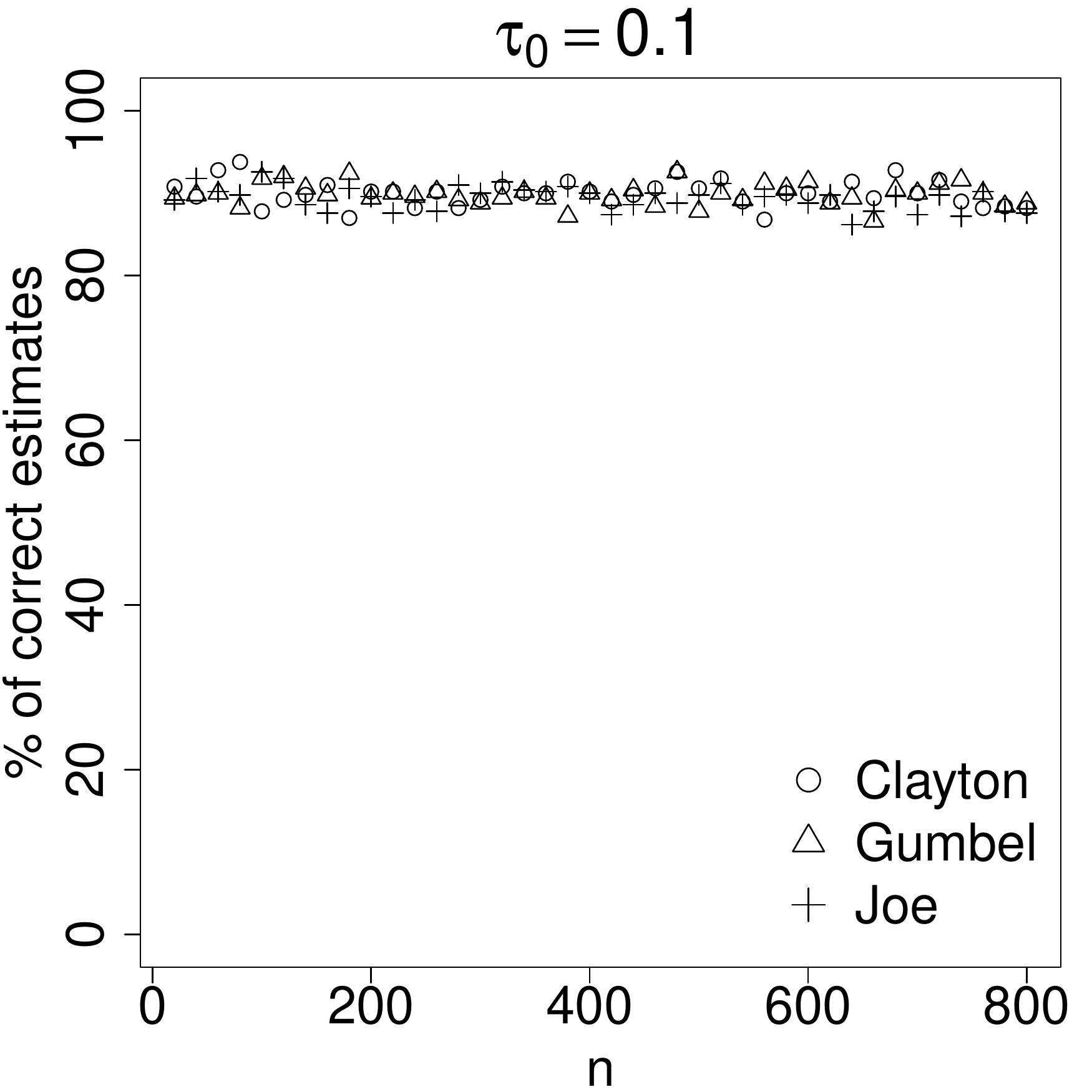}
&
\includegraphics[width=0.45\textwidth]{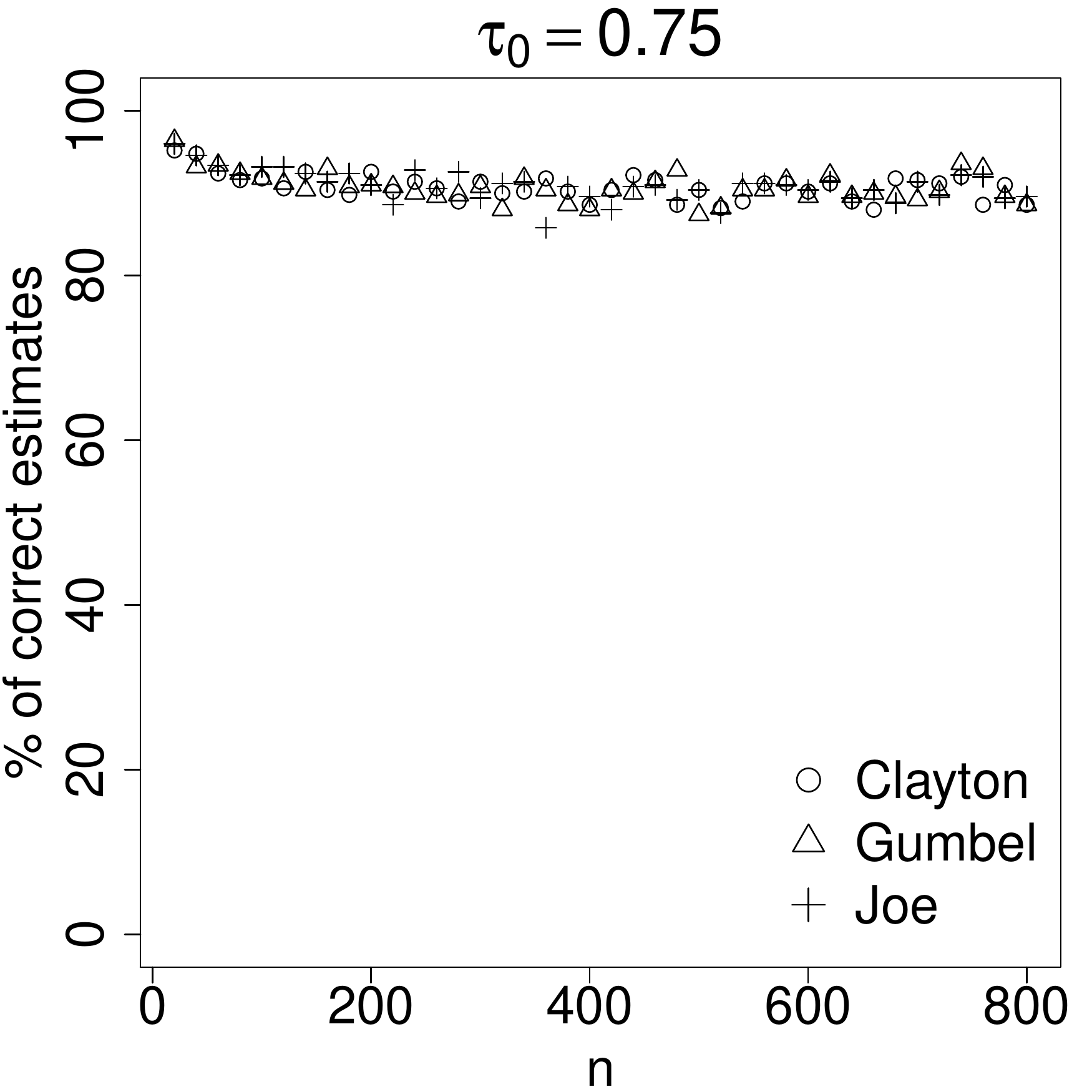}

\end{tabular}
\caption{Percentage of correct estimates when the true structure is the trivial trivariate structure. \label{AC_sim}}
\end{figure}

As expected, the percentage of correct estimates does not converge towards 100\% but oscillates around 90\%, that is, $(1-\alpha)\times 100\%$. Should we use a value for $\alpha$ of $2.5\%$ or $0.1\%$ for all $n$, the percentage of correct estimates would likewise oscillate around $97.5\%$ or $99.9\%$ respectively.

If we generate samples from a 5-fan structure (left-hand side of Figure \ref{five-variate_sim}), with $\tau_0$ arbitrarily set to 0.5 for all tested generator families and the same arbitrary value of $\alpha$ as before, the same lack of consistency can be observed, as shown on the right-hand side of Figure \ref{five-variate_sim}. Notice that the percentage of correct estimates in this case is near 100\%, even though $\alpha=0.1$. This excellent performance can be explained by the way faulty structures are handled: recall that the strategy is to decrease $\alpha$ until a valid structure emerges. At worst, $\alpha$ is decreased to $0\%$ and the estimated structure is then the trivial five-variate structure which happens to be the target structure in this case, meaning this rule of lowering $\alpha$ not only ensures that $\hat{\lambda}$ is always a proper tree but also improves the performance of our estimator for this particular case.

\begin{figure}[H]
\centering
\begin{tabular}{cc}

\includegraphics[width=0.25\textwidth]{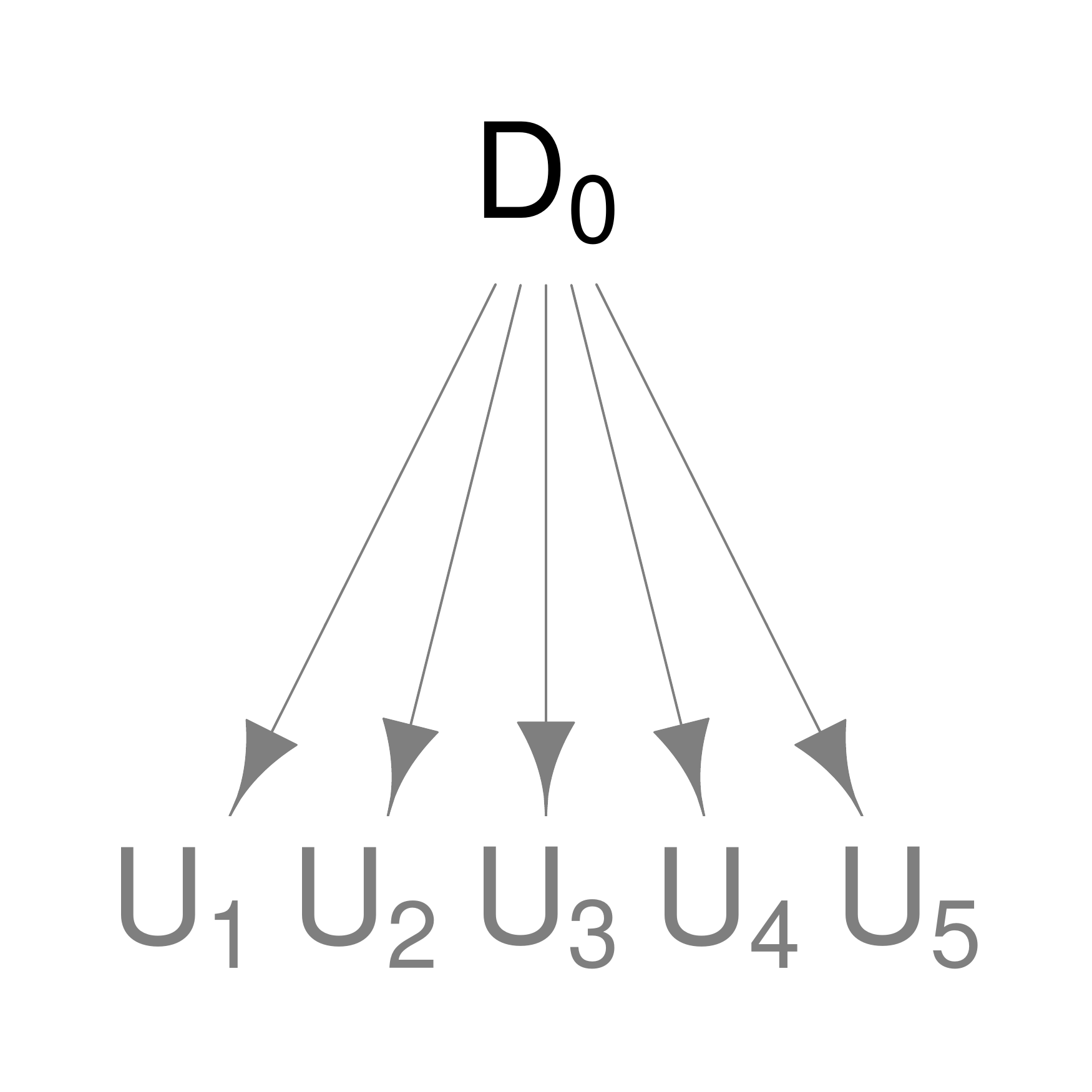}
&
\includegraphics[width=0.45\textwidth]{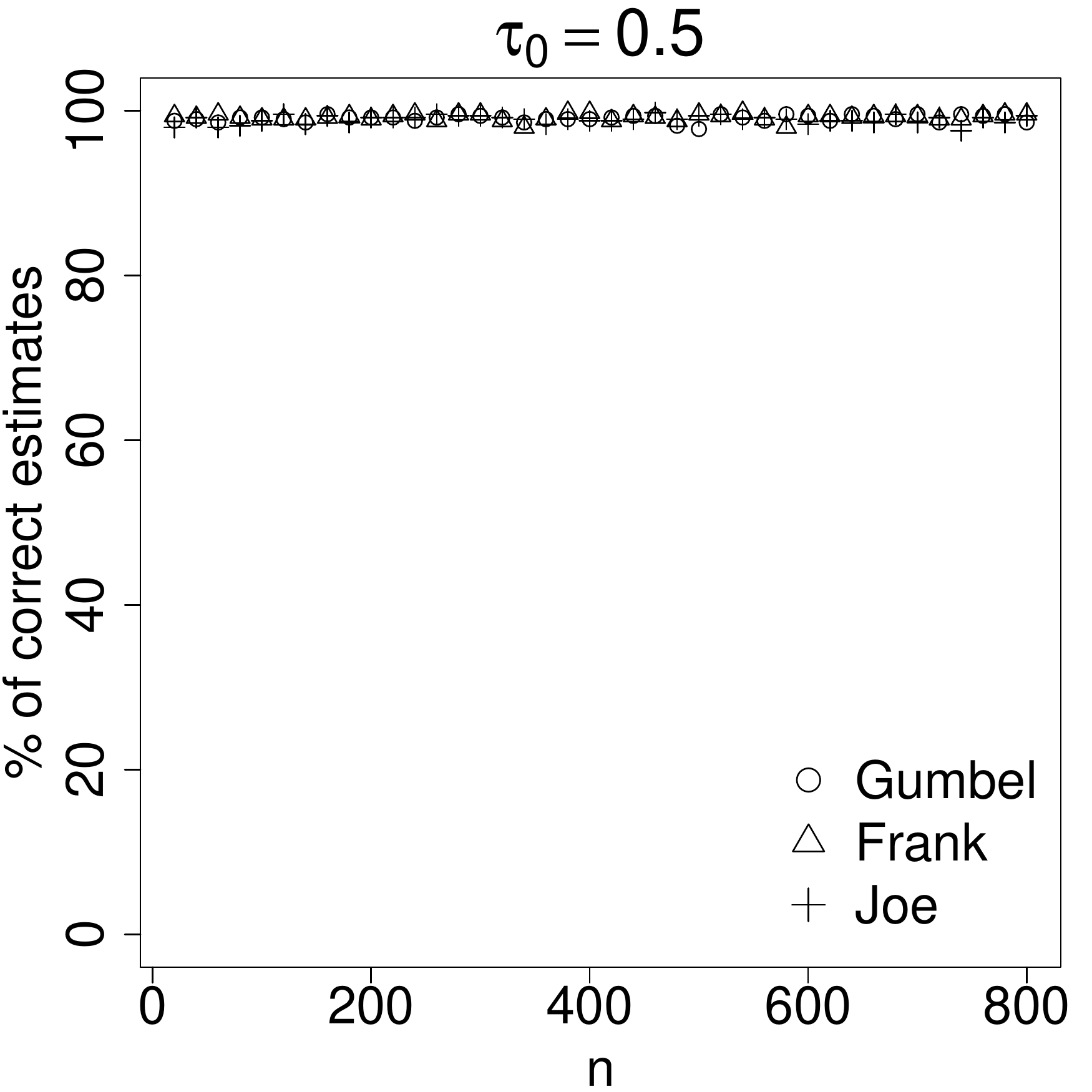}

\end{tabular}
\caption{Percentage of correct estimates for the trivial five-variate case.\label{five-variate_sim}}
\end{figure}

  

Finally, if we generate samples from the structure on the left-hand side of Figure \ref{four-variate_sim}, with $\tau_0=0.3$ and $\tau_{34}=0.7$ for all tested generator families (note that the same generator family is always used across all nodes of a given structure in the simulation section of this paper), we again can see a lack of consistency, as the percentage of correct estimates eventually oscillates around 97\% (right-hand side of Figure \ref{four-variate_sim}). Since two of the trivariate components of this structure are 3-fans, namely the structure of $(U_1, U_2, U_3)$ and the structure of $(U_1, U_2, U_4)$, this lack of consistency was, again, expected.

\begin{minipage}{\textwidth}
  \centering
  \raisebox{-0.5\height}{\includegraphics[width=0.25\textwidth]{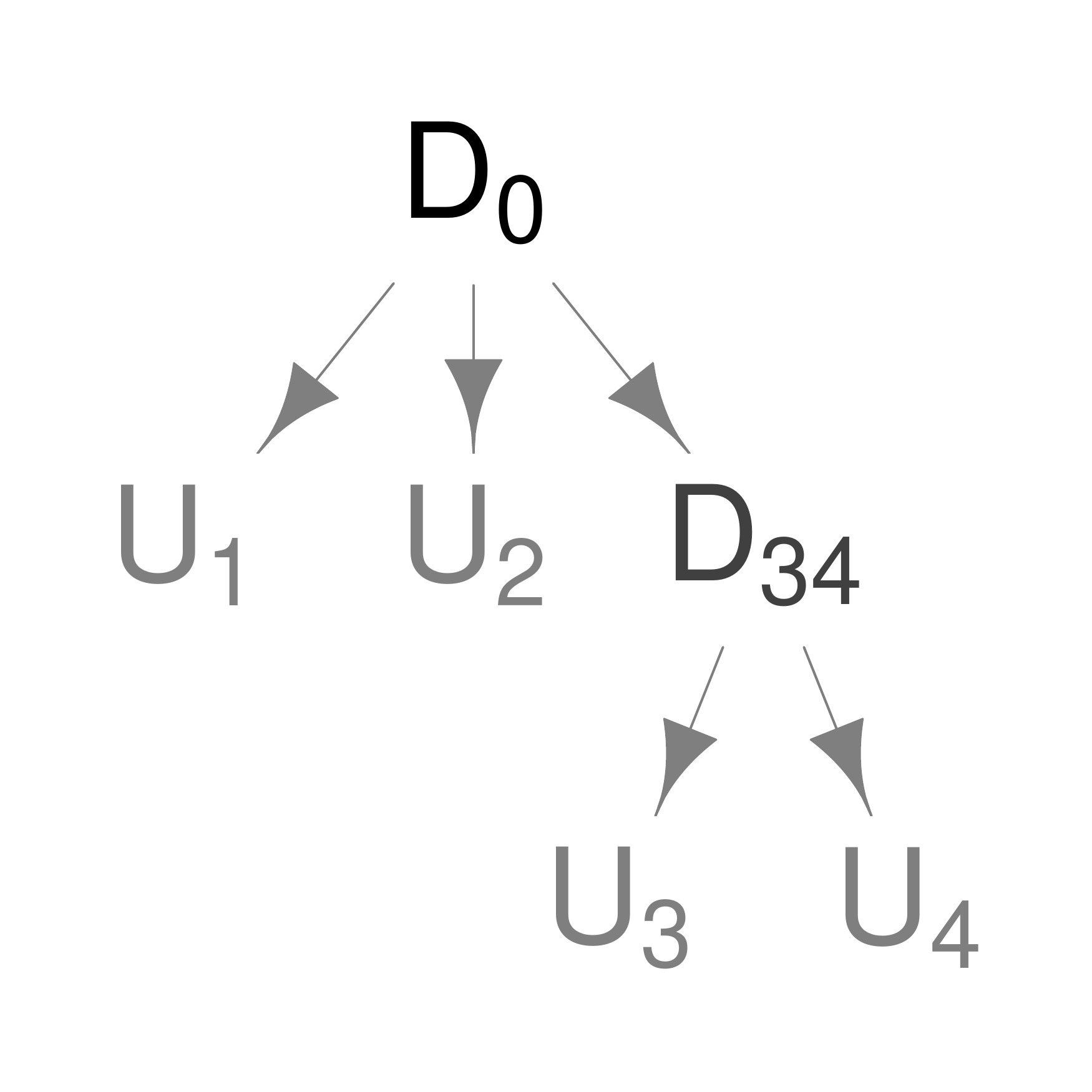}}
\hspace*{.2in}
  \raisebox{-0.5\height}{\includegraphics[width=0.45\textwidth]{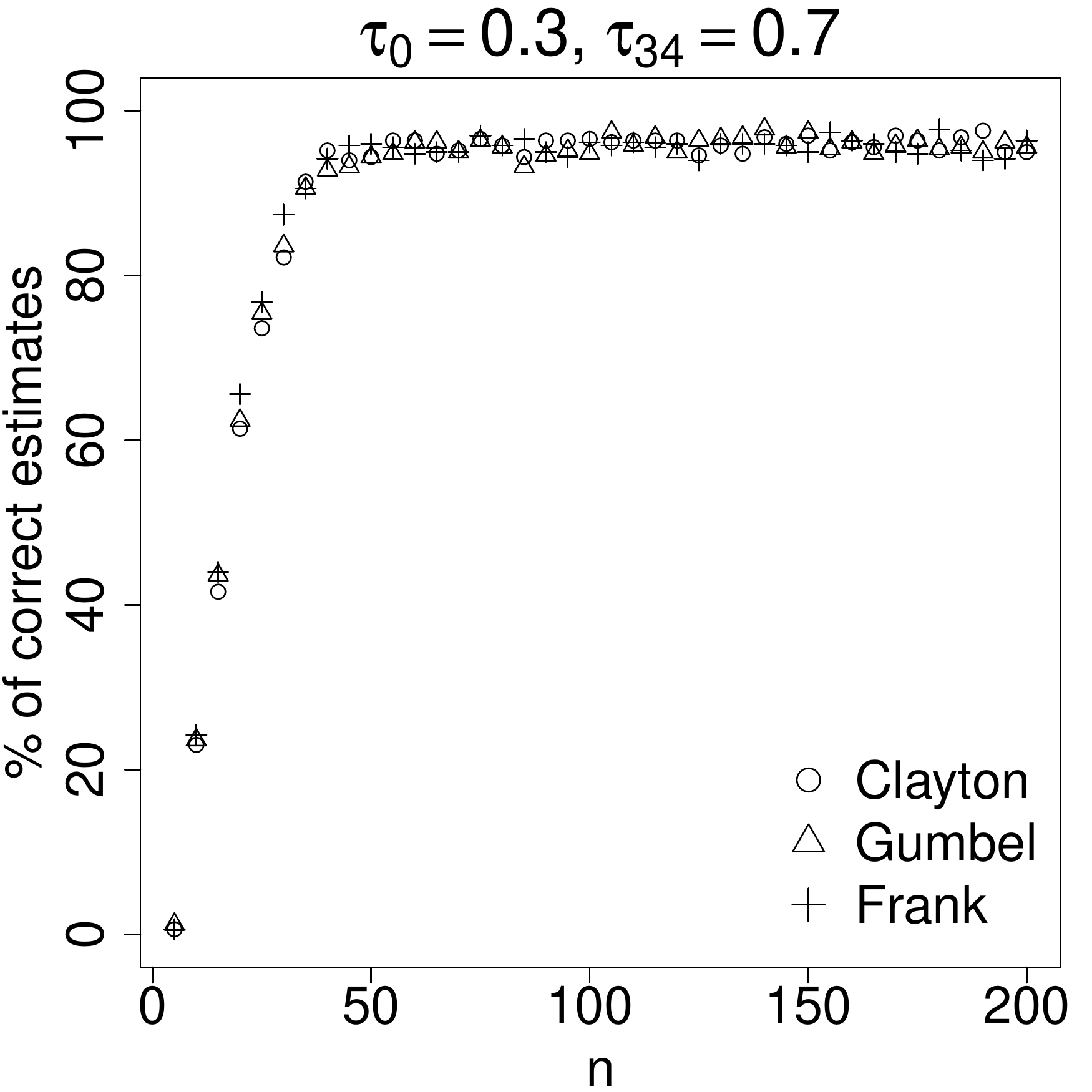}}
  
\end{minipage}
\begin{figure}[H]
\caption{Percentage of correct estimates for a four-variate case.\label{four-variate_sim}}
\end{figure}

In all these cases, it seems consistency can be achieved only if we let $\alpha$ tend to 0 as $n$ increases, in order to ensure type I errors are asymptotically impossible.

\subsection{Testing the method from \cite{OOW} with samples from a 3-fan, a 5-fan and a four-variate structure containing two 3-fans. \label{Ostapfan}}

In order to estimate a NAC structure, \cite{OOW} advise to use what they call the binary aggregated grouping with recursive estimation method, or RML method in short. Essentially, this approach consists in building a fully nested tree from bottom to top and then to aggregate some of the nodes of the resulting tree according to some criterion, so that the final estimated structure can possibly be something else than a fully nested tree.

To apply the RML method throughout the simulation section of this paper, we used the function \textbf{estimate.copula} of the \textsf{R} package HAC \citep{RePEc:hum:wpaper:sfb649dp2012-036}, this package being related to the work of \cite{OOW}. Since only the Clayton and Gumbel generator families are currently implemented in the HAC package, assessment of the RML method performance for other generator families is not possible at the time of writing. 

For the aggregation step in their approach, \cite{OOW} suggest several criteria. We used the only criterion currently implemented in the HAC package, namely that for any two successives nodes with estimated parameters $\hat{\theta}_I$ and $\hat{\theta}_J$ in the structure, the nodes have to be aggregated if $|\hat{\theta}_I-\hat{\theta}_J|<\epsilon$, where $\epsilon$ has to be chosen by the user. 

With a value for $\epsilon$ arbitrarily set to 0.30 for all $n$, Figure \ref{Ostap_sim1} displays the performance of their method for the estimation of a 3-fan.

\begin{figure}[H]
\centering
\begin{tabular}{cc}

\includegraphics[width=0.45\textwidth]{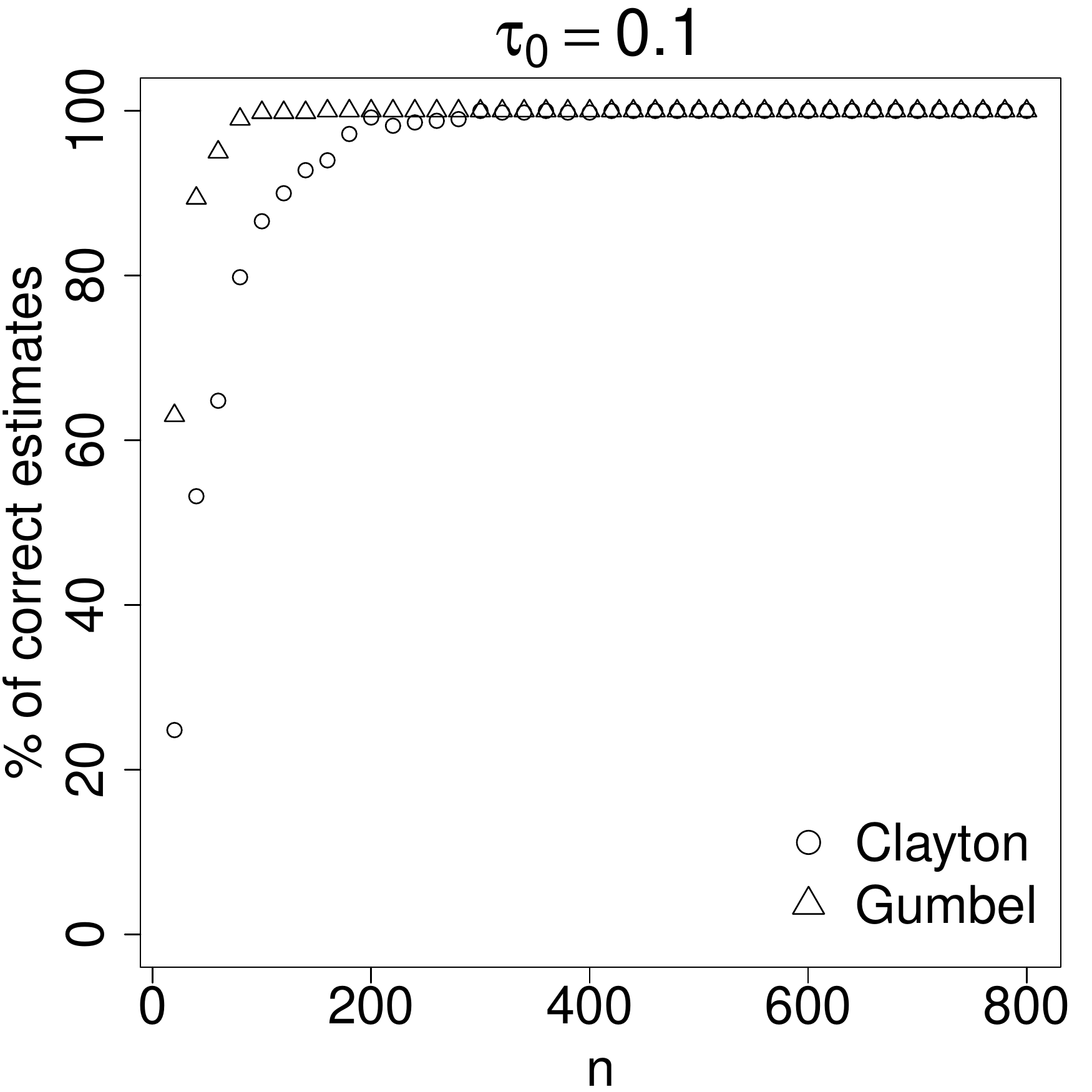}
&
\includegraphics[width=0.45\textwidth]{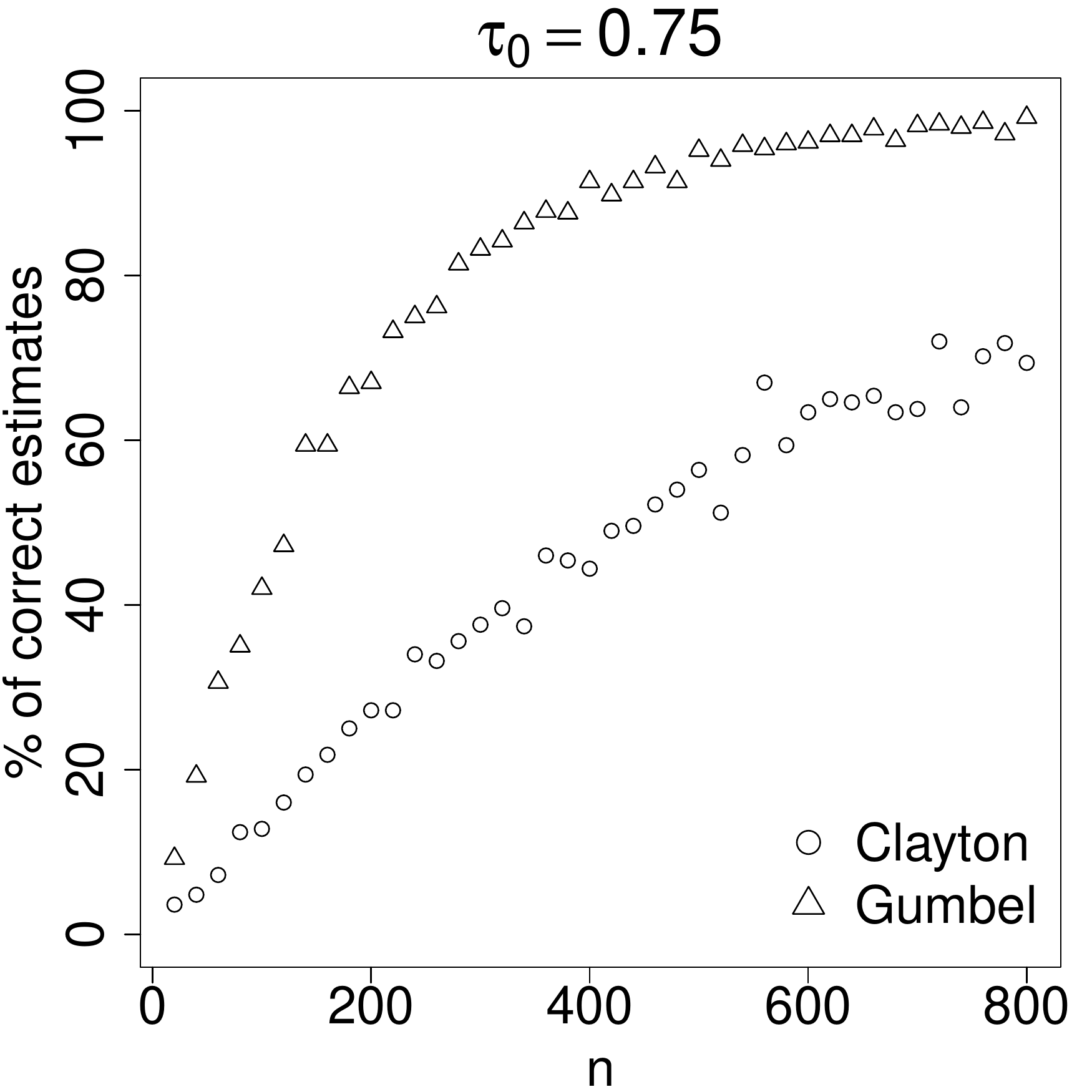}

\end{tabular}
\caption{Performance of the RML method by \cite{OOW} for a $3$-fan, with $\epsilon=0.30$ as threshold for aggregation.\label{Ostap_sim1}}
\end{figure}

Increasing the value of $\epsilon$ for all $n$ typically improves the performance of their estimator in the case of a 3-fan, as it increases the chances of aggregation. Lowering the value of $\epsilon$ typically deteriorates the performance of their estimator for this case. At the limit, with $\epsilon$ set to 0.00, no aggregation is done at all, and their estimator is unable to estimate correctly the trivial trivariate structure studied here. These remarks hold if the samples are generated from a 5-fan.

The case of the structure on the left-hand side of Figure \ref{four-variate_sim} is a little more complex to investigate. Their estimator is indeed able to consistently estimate this structure for $\epsilon$ set to 0.15, 0.30 or 0.60, but not for $\epsilon$ set to 5.00 for instance.

\subsection{The case of samples coming from a triple or from a seven-variate structure made up only of triples}

Given 500 samples of size $n$ from a non-trivial trivariate structure (a triple) such as the one in the left of Figure \ref{picture_tree} and $\alpha=0.10$, how many times among the 500 samples are we able to retrieve this triple with our method? Figure \ref{NAC_sim} shows the percentage of correct estimates for various values of $n$ and various generator families (again, note that the same generator family is always used across all nodes of a given structure in the simulation section of this paper). The parameters $\theta_0$ (root node, $D_0$) and $\theta_{23}$ (the other branching node, $D_{23}$) are expressed as Kendall's $\tau$ coefficients for convenience.

\begin{figure}[H]
\centering
\begin{tabular}{cc}

\includegraphics[width=0.45\textwidth]{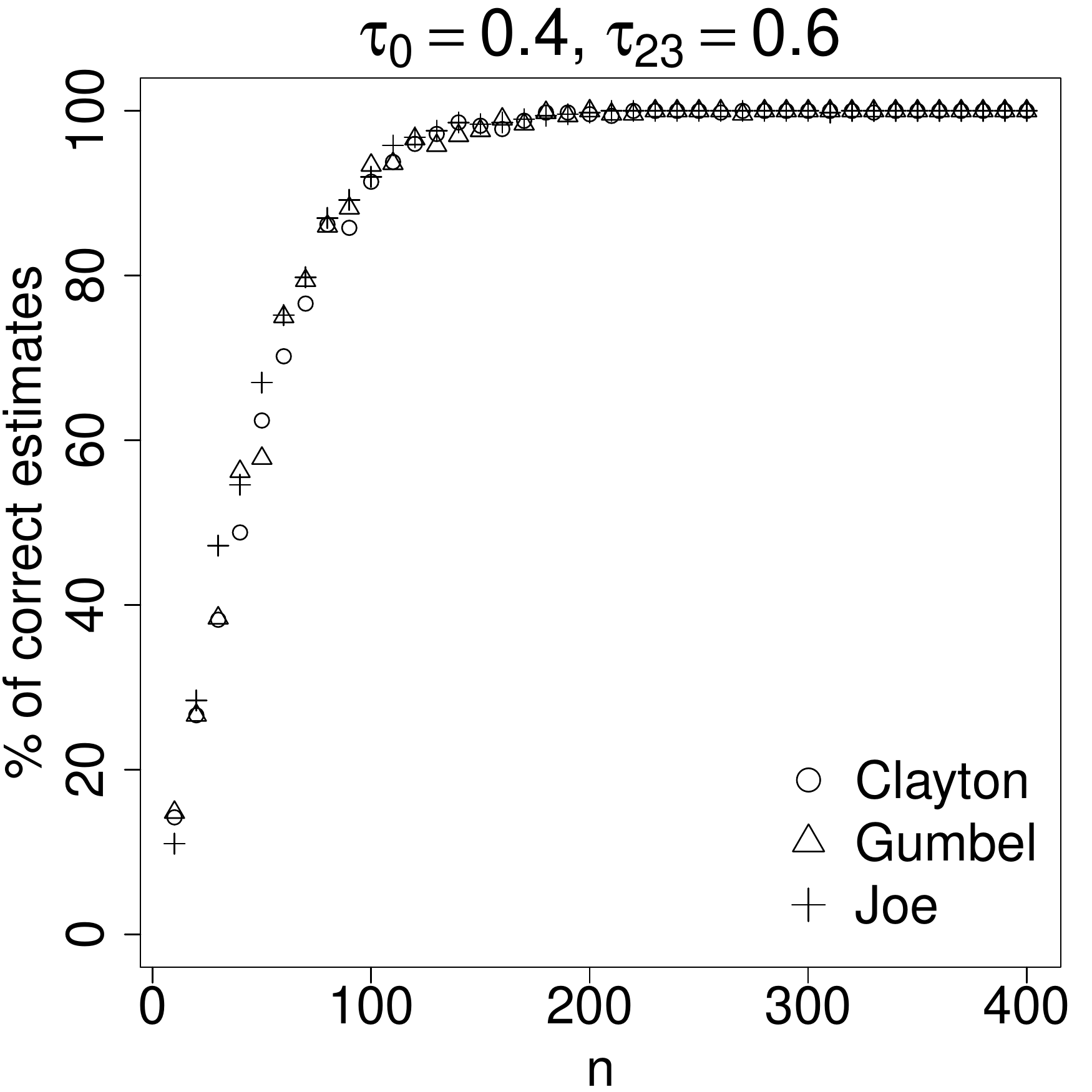}
&
\includegraphics[width=0.45\textwidth]{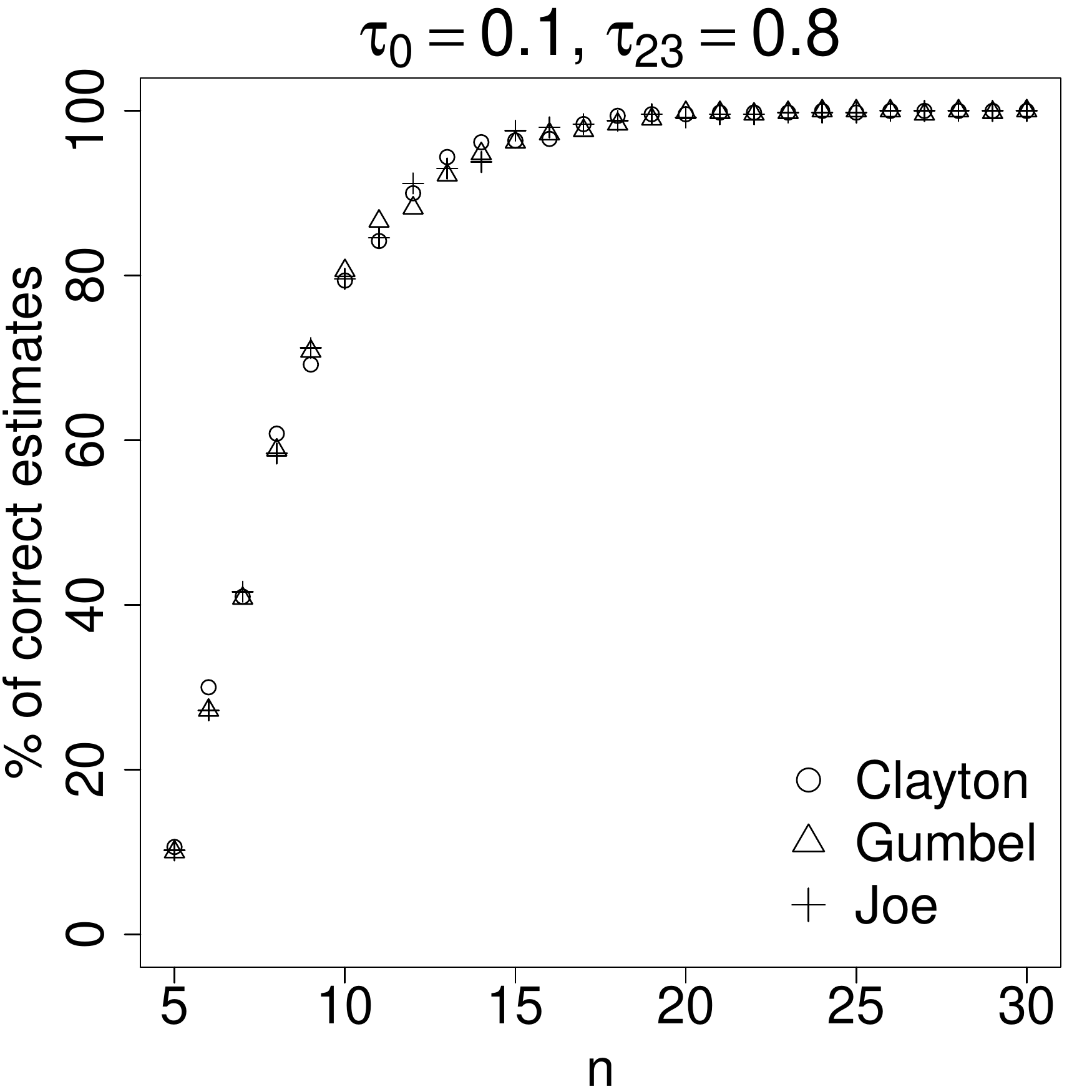}

\end{tabular}
\caption{Percentage of correct estimates when $d=3$ and true structure is $\lambda_{23}$. \label{NAC_sim}}
\end{figure}

As the sample size increases, there is a clear convergence towards 100\% of correct estimates. The more apart $\tau_0$ and $\tau_{23}$, the faster the convergence towards 100\% of correct estimates (compare the two horizontal axes in Figure \ref{NAC_sim}). These results strongly suggest our estimator, at least when $\alpha=10\%$, is a consistent estimator for any non-trivial trivariate NAC structure and thus for any larger NAC structure made up only of triples. Indeed, if the samples are generated from the seven-variate structure such as the one on the left-hand side of Figure \ref{sevenvariate_sim}, with $\tau_0=0.1$, $\tau_{123}=0.3$, $\tau_{23}=0.6$, $\tau_{4567}=0.3$, $\tau_{567}=0.5$ and $\tau_{67}=0.8$ for all tested generator families, the proportion of correct estimates grows to 100\% as $n$ increases (right-hand side of Figure \ref{sevenvariate_sim}).

\begin{minipage}{\textwidth}
  \centering
  \raisebox{-0.5\height}{\includegraphics[width=0.45\textwidth]{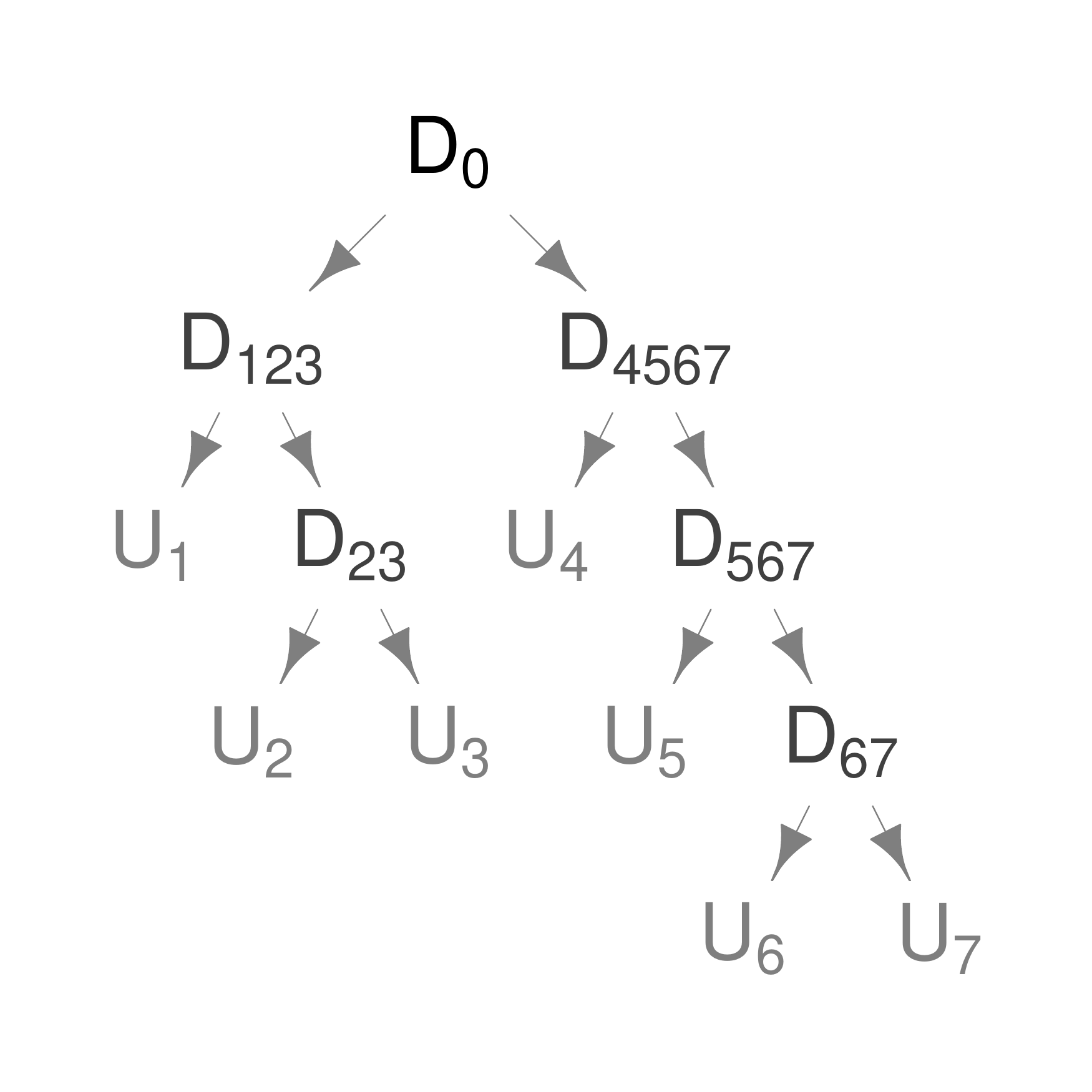}}
\hspace*{.2in}
  \raisebox{-0.5\height}{\includegraphics[width=0.45\textwidth]{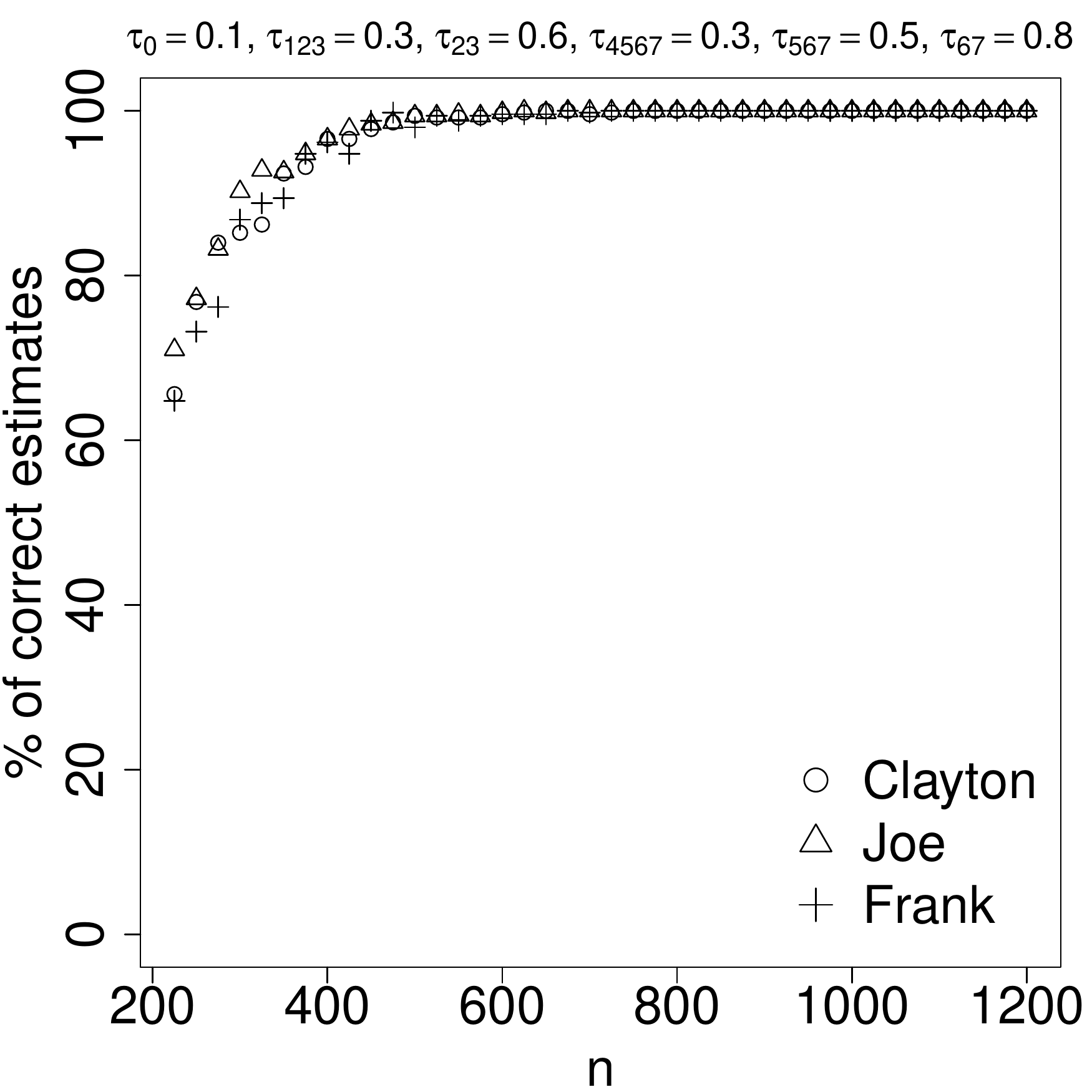}}
\begin{figure}[H]
\caption{Percentage of correct estimates (right) for a seven-variate structure (shown on the left). \label{sevenvariate_sim}}
\end{figure}
\end{minipage}

Increasing the value of $\alpha$ for all $n$ actually further improves the performance of our estimator for both structures, the best performance possible being delivered when $\alpha$ is set to its upper limit, that is, $\alpha=100\%$. If $\alpha$ is set to its lower limit, that is $\alpha=0\%$, our estimator becomes unable to estimate correctly any of the two structures studied here.

When the target structure is a triple, we found that the percentage of correct estimates also converges towards 100\% by using the RML method from \cite{OOW} as we did in Subsection \ref{Ostapfan}, provided the value of $\epsilon$ is small enough. In fact, as any aggregation should be avoided in case the target structure is a triple, the performance of their estimator typically improves by lowering $\epsilon$ for all $n$, the best performance being delivered when $\epsilon=0$, that is, when the aggregation step is completely skipped. Should the value of $\epsilon$ be too large, then their estimator will fail to be a consistent estimator for the non-trivial trivariate structure.

In case the target structure is a triple, Figure \ref{limit} allows for a direct comparison between our approach and their approach when both are pushed to their favorable respective limit, thus with $\alpha=100\%$ for our method and with $\epsilon=0$ for the method of \cite{OOW}.

\begin{figure}[H]
\centering
\begin{tabular}{cc}

\includegraphics[width=0.45\textwidth]{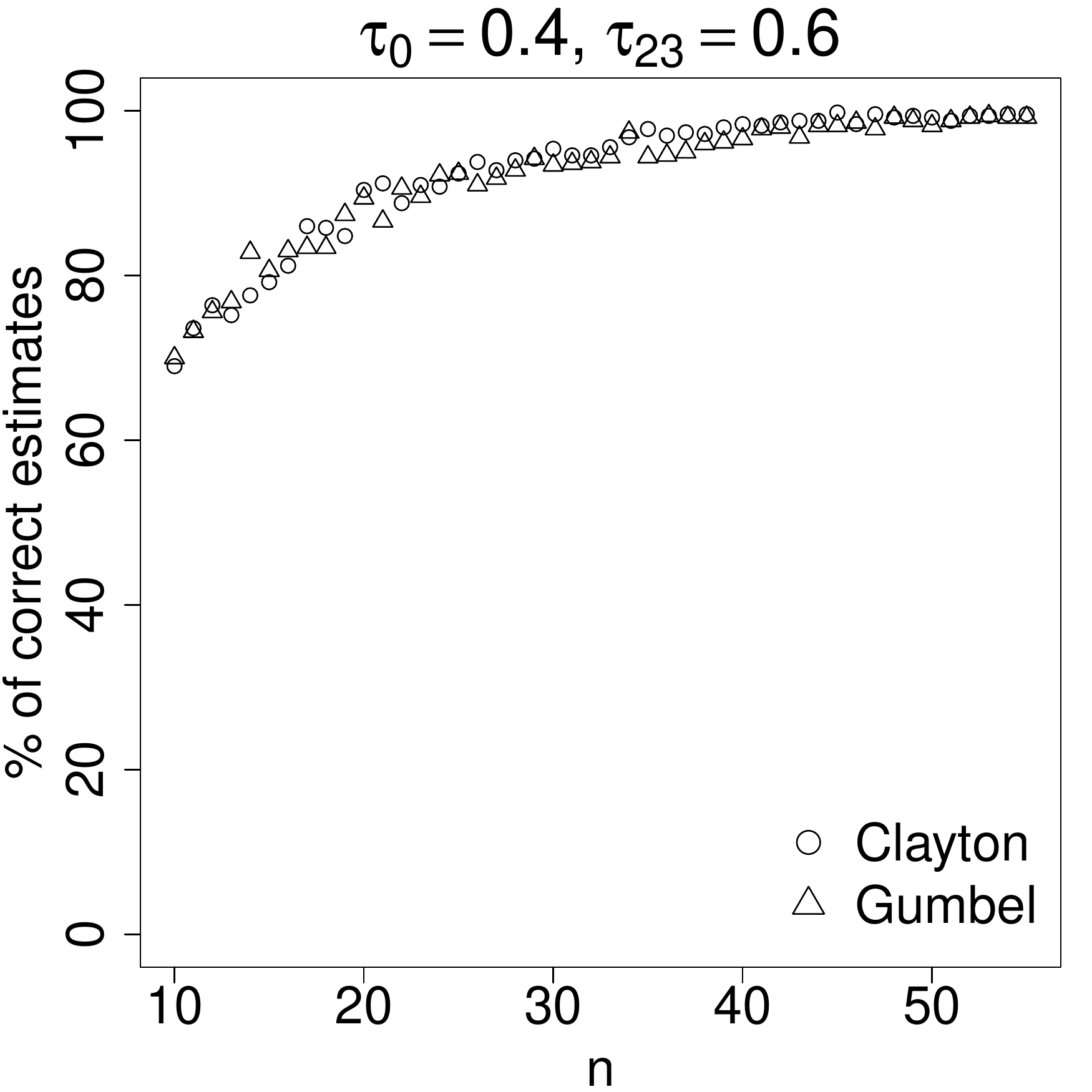}
&
\includegraphics[width=0.45\textwidth]{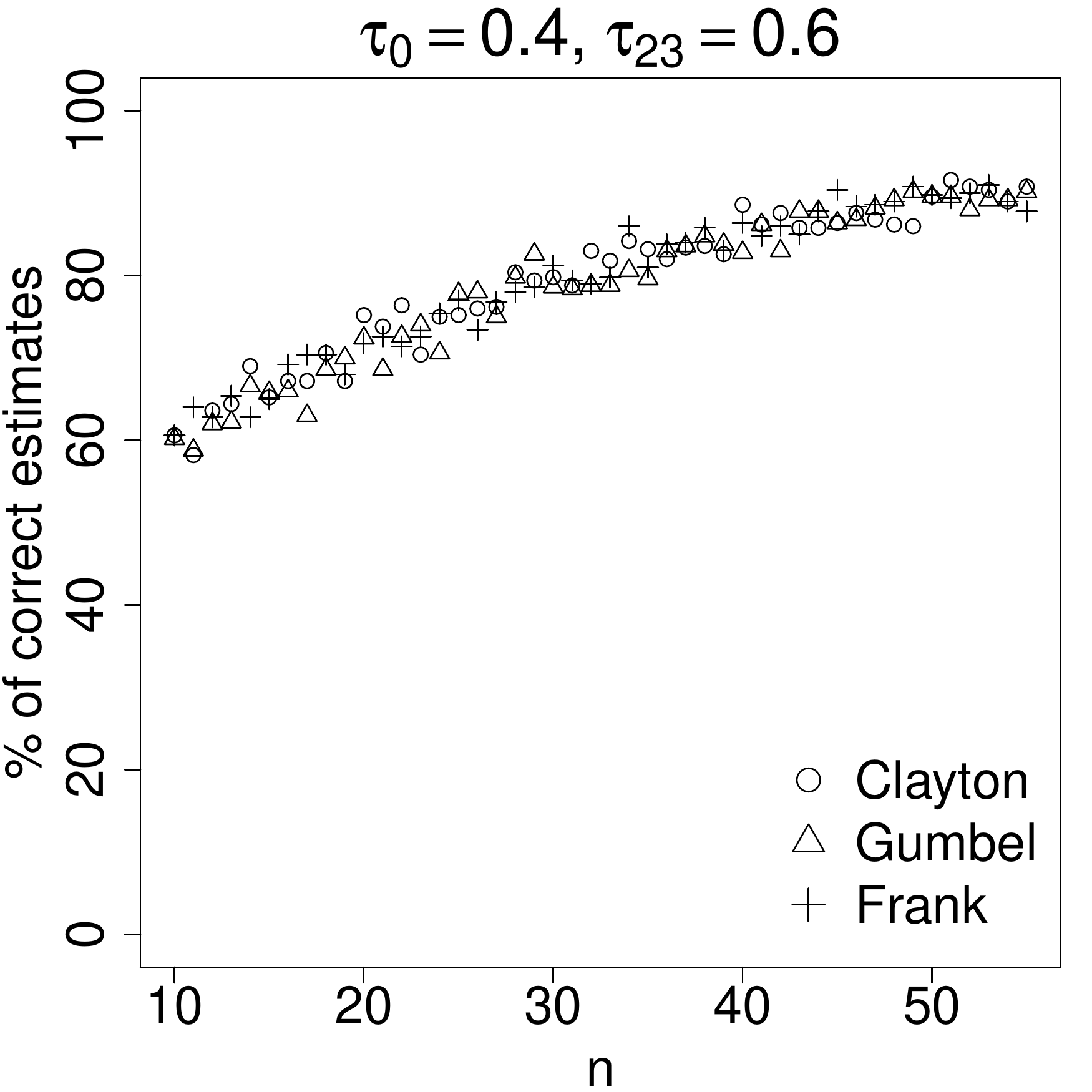}

\end{tabular}
\caption{Percentage of correct estimates when $d=3$ and true structure is $\lambda_{23}$. Left is the RML method for Clayton and Gumbel generators, right is the method described in this paper for Clayton, Gumbel and Frank generators, the latter generator being an arbitrary choice. Both methods were pushed to their respective limit in order to deliver the best performance possible for structure $\lambda_{23}$. \label{limit}}
\end{figure}

There is a performance gap between the two methods. Recall however that the RML method was applied with the prior knowledge that the generators were Clayton and Gumbel generators while our method does not require such prior knowledge.

\section{Application \label{app_sec}}
Daily log returns from January 2010 to December 2012 of the following indices were gathered with the help of Yahoo! Finance: 
\begin{itemize}
  \setlength{\itemsep}{1pt}
  \setlength{\parskip}{0pt}
  \setlength{\parsep}{0pt}
\item Abercrombie \& Fitch Co. (ANF), traded in New York;
\item Amazon.com Inc. (AMZN), traded in New York;
\item China Mobile Limited (ChM), traded in Hong Kong;
\item PetroChina (PCh), traded in Hong Kong;
\item Groupe Bruxelles Lambert (GBLB), traded in Brussels;
\item and KBC Group (KBC), traded in Brussels.
\end{itemize}
For each of these six time series, we fitted a GARCH(1,1) model with generalized error distribution and extracted the residuals, that is, we divided each of the six original time series by the related estimated standard deviations, leading to a table of $n = 740$ observations and $d = 6$ columns. These new six time series we will call the GARCH(1,1)-standardized log returns. We performed a Ljung-Box test (lag 20) on each of the six GARCH(1,1)-standardized log returns as well as on each of the six squared GARCH(1,1)-standardized log returns and failed to reject the null hypothesis of zero serial correlation each time. A chi-squared test was also performed on each of the six GARCH(1,1)-standardized log returns to check if the generalized error distribution assumed for the residuals in the GARCH(1,1) model is warranted. Again, we failed to reject the null hypothesis each time.

Figure \ref{therest} shows the estimated structure for the GARCH(1,1)-standardized log returns of ANF, AMZN, ChM and PCh, the estimated structure for ANF, AMZN, GBLB and KBC, and the estimated structure for ChM, PCh, GBLB and KBC.

\begin{figure}[H]
\centering
\begin{tabular}{ccc}

\includegraphics[width=0.3\textwidth]{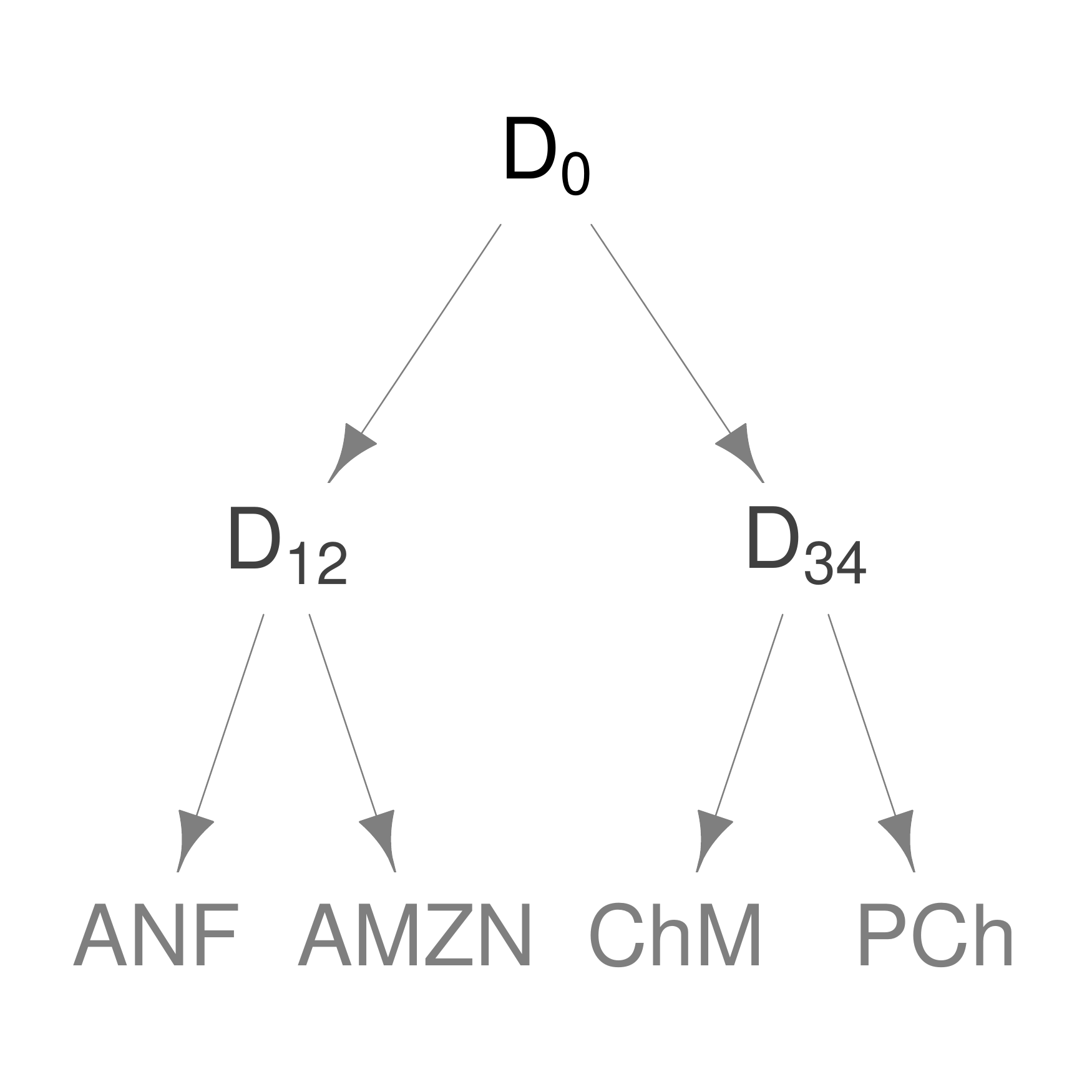}
&
\includegraphics[width=0.3\textwidth]{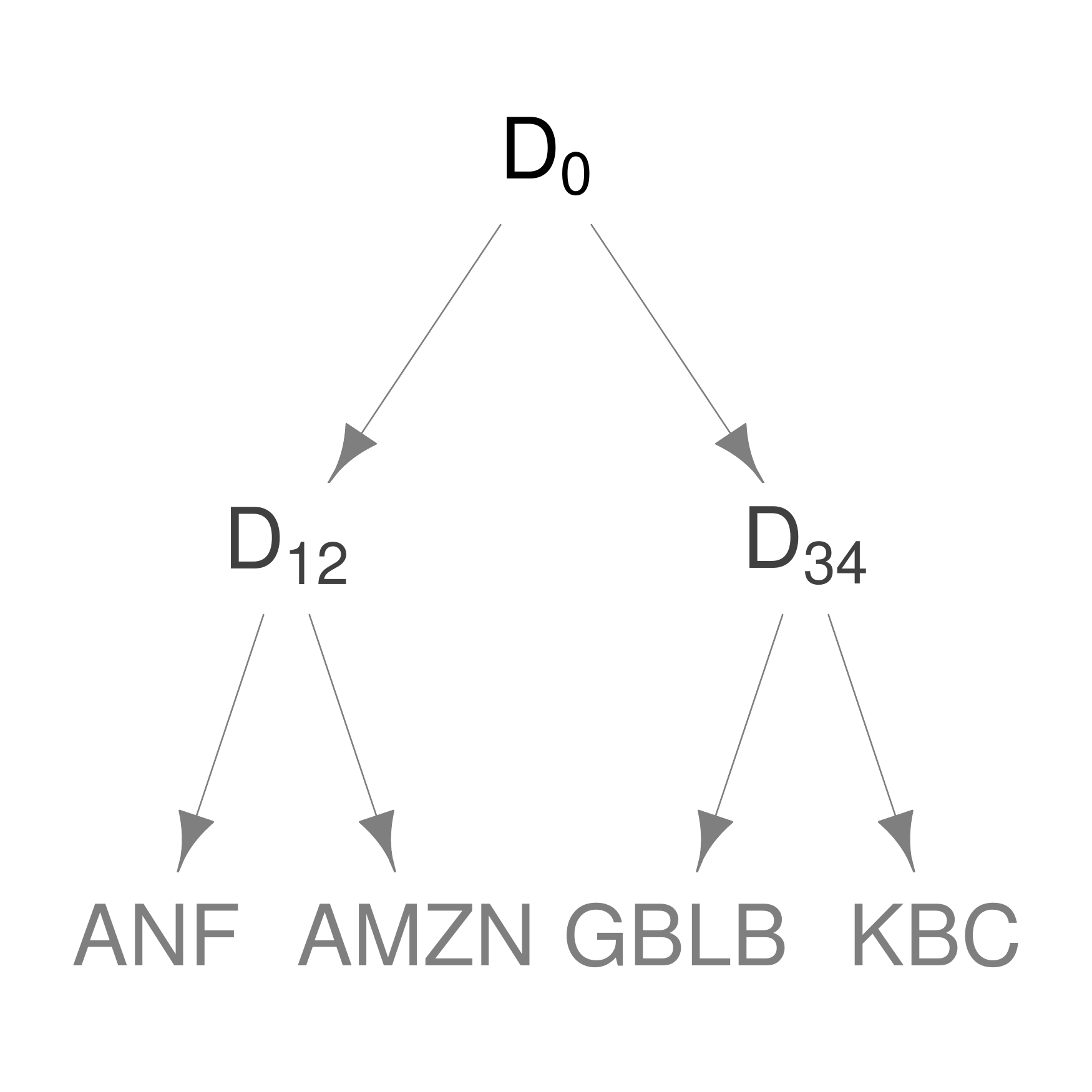}
&
\includegraphics[width=0.3\textwidth]{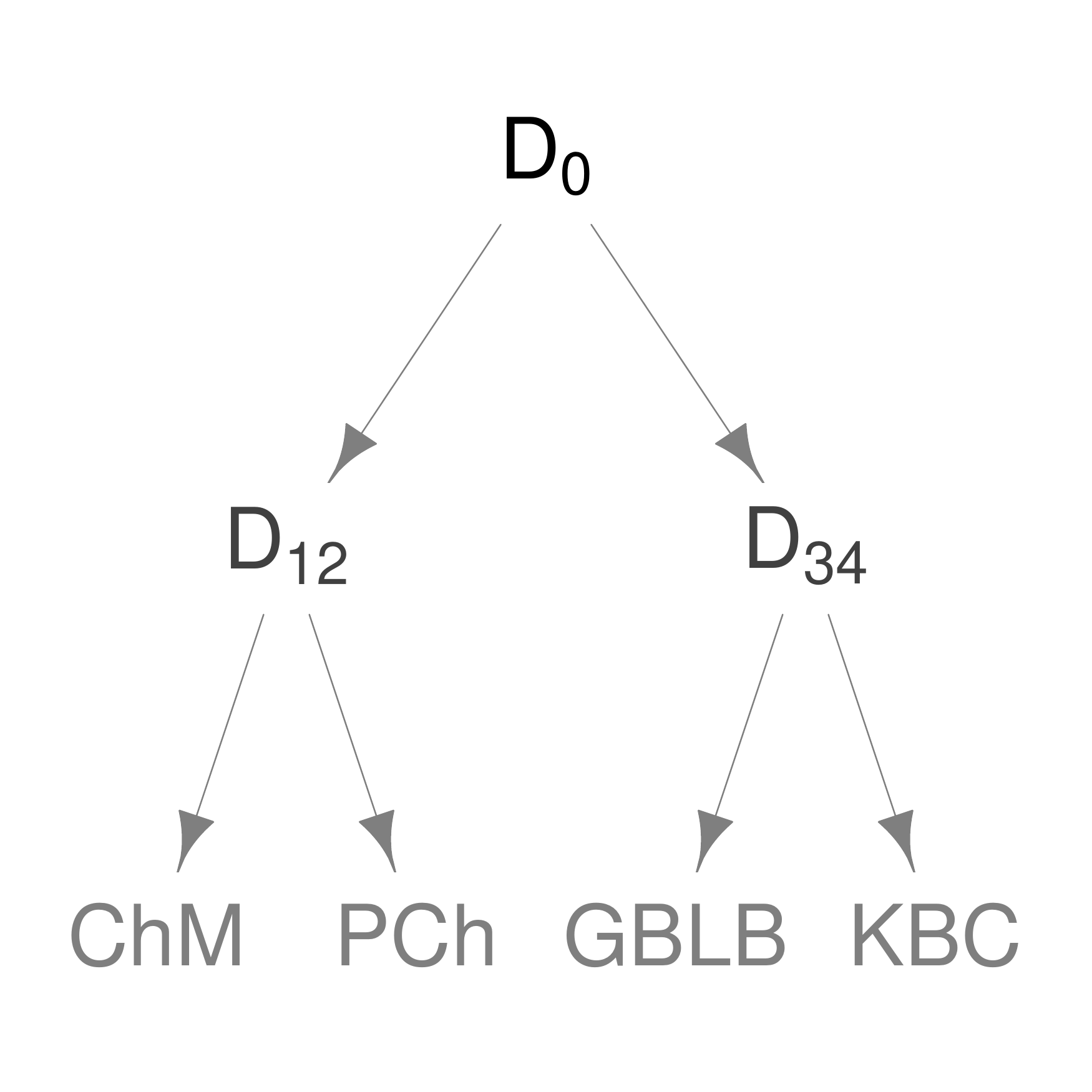}

\end{tabular}
\caption{Given two series of GARCH(1,1)-standardized log returns from one geographical area and two from another area, a natural clustering by area arises. The above structures are all strongly supported by the data, as the 12 related p-values are less than 10e-04. \label{therest}}
\end{figure}

In order to build a six-variate structure, we need to estimate eight extra trivariate structures. The left-hand side of Figure \ref{wanted} shows a reasonable guess for the six-variate structure in which the eight extra trivariate structures all are 3-fans.

\begin{figure}[H]
\centering
\begin{tabular}{cc}

\includegraphics[width=0.49\textwidth]{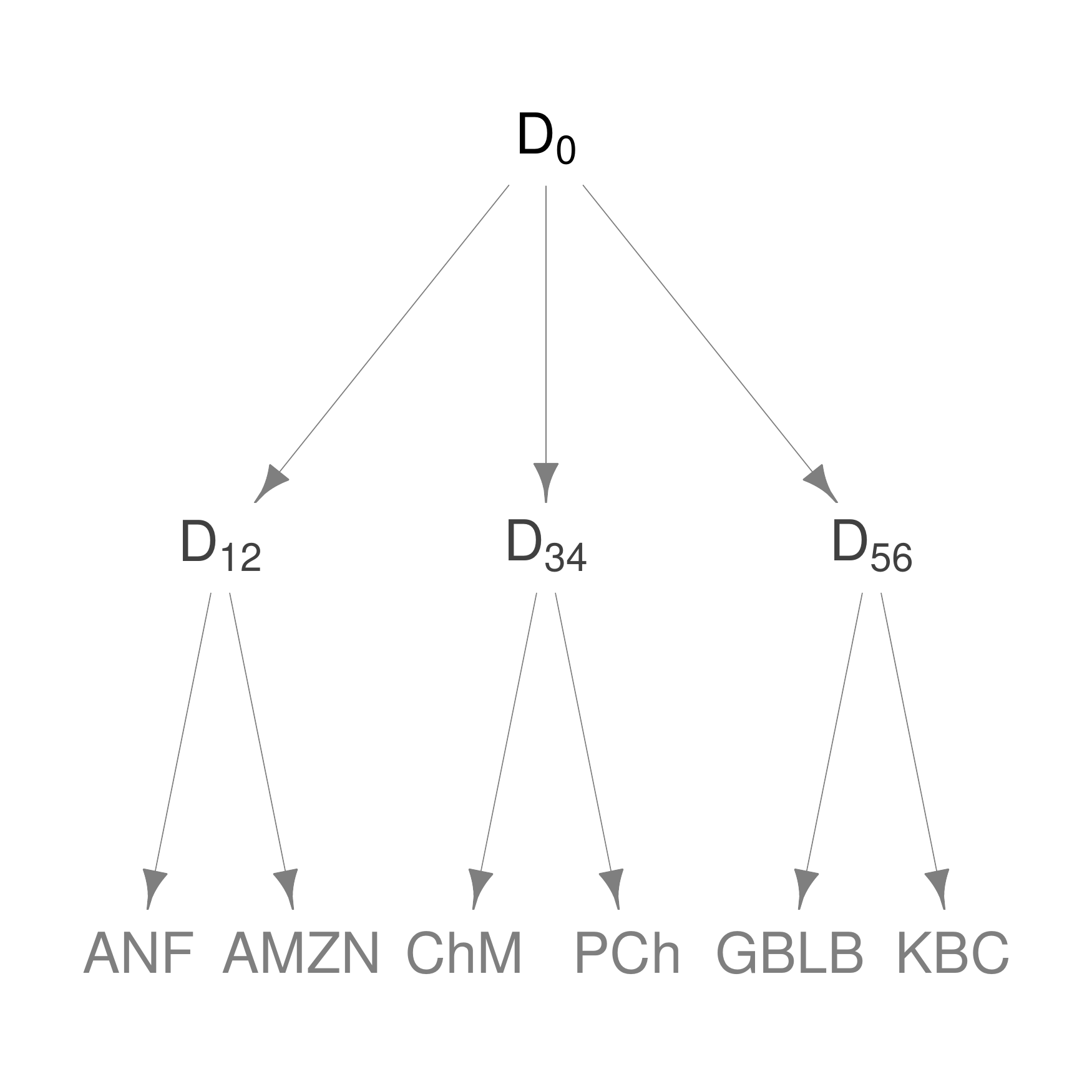}
&
\includegraphics[width=0.49\textwidth]{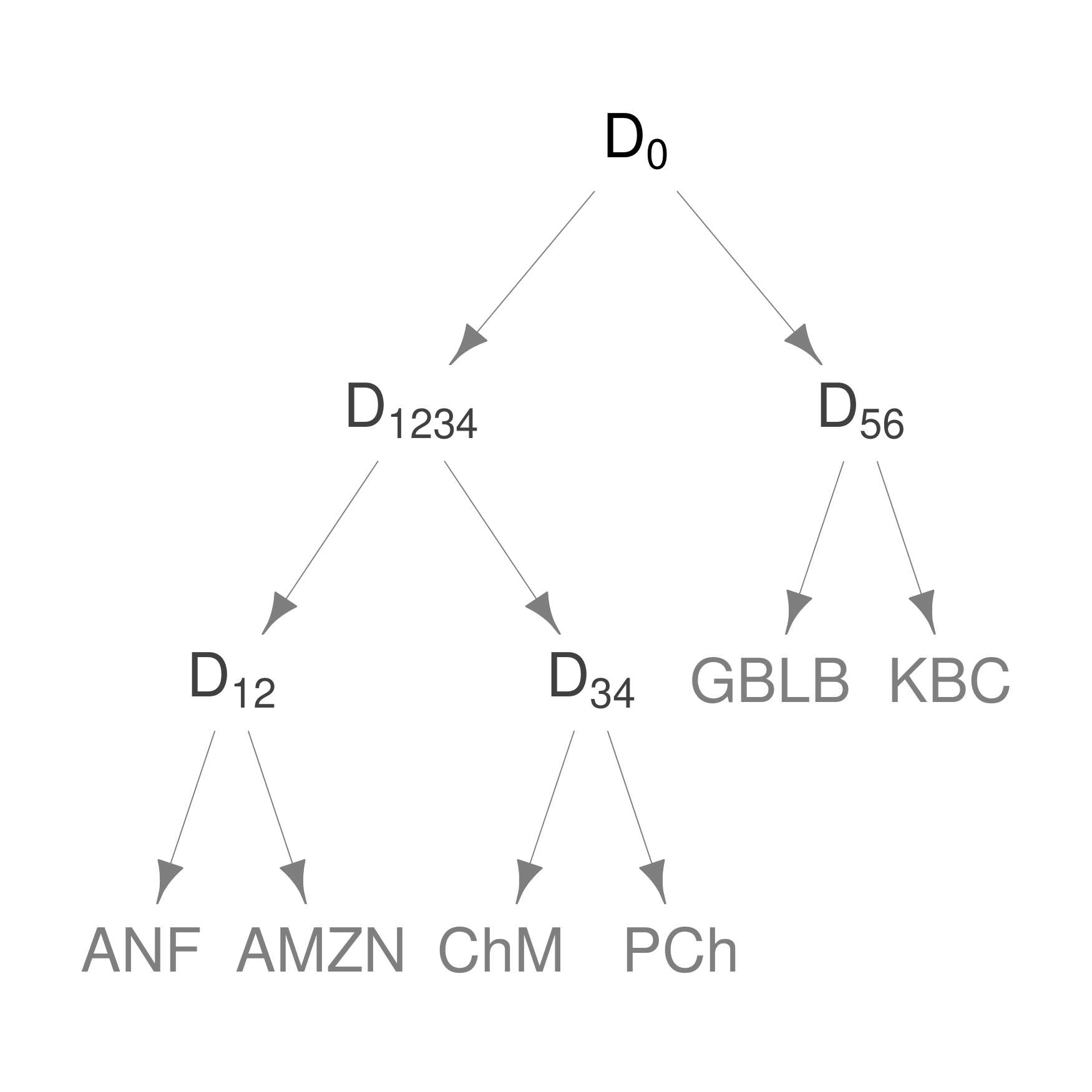}

\end{tabular}
\caption{Possible six-variate structures for the data. \label{wanted}}
\end{figure}

However, the 3-fan in four of the eight extra trivariate structures is strongly rejected by the data, which rather suggest the structure on the right-hand side of Figure~\ref{wanted}. Unfortunately, this last structure implies we must reject the 3-fan for all eight extra trivariate structures and not only for half of them, making the estimation of a six-variate structure quite uncertain. Since both PetroChina and China Mobile are traded not only in Hong Kong but also in New York, we could expect their log returns in Hong Kong to be more related to the log returns of some companies in New York (for instance ANF and AMZN) than to the log returns of two companies in Belgium. The structure on the right-hand side of Figure~\ref{wanted} seems therefore more appropriate.

\section{Discussion \label{dis_sec}}
In this paper, we have paved the way for a nonparametric rank-based approach to estimating a NAC structure, without the need to make any assumptions about the generators of the nested Archimedean copula prior to estimation of its structure apart from a natural identifiability condition. A number of challenges remain however:

\begin{itemize}
\item Difficulties can appear when the method is applied to real data for which the true copula is not necessarily a NAC. For instance, one can end up with a subset of estimated triples each strongly supported by the data (that is, very small p-values, meaning type I or type III errors are unlikely) and yet these triples contradict each other in the sense that no global structure can be retrieved unless $\alpha$ is set to 0.00 and the global estimated structure is a fan, i.e., and Archimedean copula.
\item The whole method is computationally intensive, unlike the method from \cite{OOW}. This is best understood by calculating the number of trivariate structures for which a test is necessary: to get an estimate for a five-variate structure for instance ($d=5$), we need to estimate 10 trivariate structures. With $d=10$, we have to estimate 120 trivariate structures. Regarding the estimation of a single trivariate structure, the required time depends mainly on the sample size and on the number of bootstrap replications. With 200 bootstrap replications (the value we used throughout the simulation section of this paper), a few seconds are needed at worst to get a trivariate estimated structure. Optimized \textsf{R} code is available from the authors.
\item Once a genuine NAC structure has been estimated with our nonparametric approach, the problem of estimating the generators remains. These generators cannot be estimated marginally, as doing so does not guarantee that the resulting function will be a proper copula.
\end{itemize}

\section*{Acknowledgements}
This research is supported by contract ``Projet d'Actions de Recherche Concert\'ees'' No.\ 12/17-045 of the ``Communaut\'e fran\c{c}aise de Belgique'' and by IAP research network grant nr. P7/06 of the Belgian government (Belgian Science Policy).

We are also grateful to Alexander McNeil (Heriot-Watt University) for careful reading of parts of our manuscript and for constructive, detailed feedback.
\section*{References}
\bibliographystyle{elsarticle-harv}
\bibliography{bibnac}

\section*{Appendix}
\subsection*{Proof of Lemma \ref{lem:scaspaneq}}

The proof is built in two steps. First we need to prove that, for $\varnothing \ne A \subset C \subset D_0$, we have
\[
  \lca(A, \lambda \sqcap C) = \lca(A, \lambda) \cap C.
\]

By definition, we have
\[
  \lca(A, \lambda) \cap C 
  = \biggl( \bigcap_{B \in \lambda : A \subseteq B} B \biggr) \cap C
  = \bigcap_{B \in \lambda : A \subseteq B} (B \cap C).
\]
Since $A$ is a subset of $C$ and since $A$ must be a subset of $B$, notice that requiring $A \subset B$ is equivalent to requiring $A \subset B \cap C$. Thus we can write

\[
  \lca(A, \lambda) \cap C 
  = \bigcap_{B \in \lambda : A \subseteq B \cap C} (B \cap C).
\]

On the other hand,
\[
  \lca(A, \lambda \sqcap C)
  = \bigcap_{B' \in \lambda \sqcap C : A \subseteq B'} B'.
\]
Since $\lambda \sqcap C = \{ B \cap C : B \in \lambda \} \setminus \{ \varnothing \}$ by definition, we can rewrite the above expression as

\[
  \lca(A, \lambda \sqcap C)
  = \bigcap_{B \in \lambda : A \subseteq B \cap C, B \cap C \neq \varnothing} (B \cap C).
\]
And because $A \subseteq B \cap C$ and $A \neq \varnothing$, the requirement $B \cap C \neq \varnothing$ can be dropped, thus

\[
  \lca(A, \lambda \sqcap C)
  = \bigcap_{B \in \lambda : A \subseteq B \cap C} (B \cap C)=\lca(A, \lambda) \cap C.
\]

The second step of the proof of Lemma \ref{lem:scaspaneq} begins by making use of the result from the first step. Indeed, we can now write:
\[
  \lca(T_j, \lambda \sqcap C) = \lca(T_j, \lambda) \cap C \text{ with } j = 1, 2.
\]
Suppose to begin $\lca( T_1, \lambda ) = \lca( T_2, \lambda )$. We therefore have 
\begin{align*}
  \lca(T_1, \lambda \sqcap C) &= \lca(T_1, \lambda) \cap C \\
  &= \lca(T_2, \lambda) \cap C \\
  &= \lca(T_2, \lambda \sqcap C).
\end{align*}

On the other hand, suppose that $\lca( T_1, \lambda \sqcap C ) = \lca( T_2, \lambda \sqcap C )$. Obviously,
\[
  \lca( T_1, \lambda ) \supset \lca( T_1, \lambda ) \cap C,
\]
and since $T_2$ is both a subset of $\lca( T_2, \lambda )$ and of $C$, we also have
\[
  \lca( T_2, \lambda ) \cap C \supset T_2.
\]
Because $\lca( T_1, \lambda \sqcap C ) = \lca( T_2, \lambda \sqcap C )$ implies that $\lca( T_1, \lambda ) \cap C = \lca( T_2, \lambda ) \cap C$, we have
\[
  \lca( T_1, \lambda ) \supset T_2,
\]
which means that $\lca( T_1, \lambda )$ is an ancestor of $T_2$, but not necessarily the lowest. Therefore $\lca( T_1, \lambda ) \supset \lca( T_2, \lambda )$. The converse inclusion holds as well, by symmetry of the argument. We conclude that the two sets $\lca( T_1, \lambda )$ and $\lca( T_2, \lambda )$ are in fact equal.

\subsection*{Proof of Lemma \ref{lem:sca}}

Suppose first that $A$ is the lowest common ancestor of $B$. Clearly $B \subset A$. Let $B_1, \ldots, B_p$ be the children of $A$ and recall these children form a partition of $A$. Hence $B = B \cap A = \bigcup_{j = 1}^p (B \cap B_j)$, and thus at least one of these intersections is not empty. However, if only one of these intersections would be nonempty, say $B \cap B_1$, then we would get $B = B \cap B_1$ and thus $B \subset B_1$, meaning that $B_1$ is also common ancestor of all elements of $B$. Since $B_1$ is a proper subset of $A$, this would be in contradiction with the assumption that $A$ is the lowest common ancestor of $B$. Therefore if $A$ is the sca of $B$, $B$ has a nonempty intersection with a least two children of $A$.

Conversely, suppose that $B \subset A$ and that there exist distinct children $B_1$ and $B_2$ of $A$ having nonempty intersections with $B$. Let $A'$ be a node in $\lambda$ such that $B \subset A'$. Then also $B \cap B_1 \subset A'$, and thus, as $B \cap B_1$ is nonempty, $A' \cap B_1$ is not empty. Similarly, $A' \cap B_2$ is not empty. Since $B_1$ and $B_2$ are disjoint, requirement (iii) in Definition~\ref{def:tree} then forces $B_1$ and $B_2$ to be descendants of $A'$. As a consequence $A \subset A'$. We have obtained that $A$ is included in every node $A'$ containing $B$ as a subset. We conclude that $A$ is the lowest common ancestor of the elements of $B$, as required.

\end{document}